\documentclass[10pt,a4paper]{article}

\author{%
  Richard A. Davis\\
  \small Columbia University\\
  \small Department of Statistics\\
  \small 1255 Amsterdam Ave. New York, NY 10027, USA\\
  \small \texttt{rdavis@stat.columbia.edu}
  \and
  Holger Drees\\
  \small University of Hamburg\\
  \small Department of Mathematics\\
  \small Bundesstra{\ss}e 55, 20146 Hamburg, Germany\\
  \small \texttt{holger.drees@math.uni-hamburg.de}
  \and
  Johan Segers \qquad
  Micha\l~Warcho\l\\
  \small Universit\'e catholique de Louvain\\
  \small Institut de Statistique, Biostatistique et Sciences Actuarielles\\
  \small Voie du Roman Pays 20, B-1348 Louvain-la-Neuve, Belgium\\
  \small \texttt{johan.segers@uclouvain.be}, \texttt{m87warchol@gmail.com}
}

\title{Inference on the tail process with application to financial time series modelling}

\date{\today}

\usepackage[utf8]{inputenc}
\usepackage[authoryear]{natbib}
\usepackage{amsmath}
\usepackage{amsfonts}
\usepackage{amsthm}
\usepackage{amssymb}
\usepackage{bm}
\usepackage[colorlinks=true,linkcolor=blue,citecolor=blue,pdfborder={0 0 0}]{hyperref}
\usepackage{paralist}
\usepackage{graphicx}
\usepackage{booktabs}
\usepackage{enumerate}
\usepackage{a4wide}
\usepackage{color}
\usepackage{url}
\usepackage{soul}

\newcommand{\Ft}{F^{(\Theta_t)}}

\newcommand{\CDFfT}[1]{\ensuremath{\hat{F}^{(\mathrm{f},\Theta_t)}_{n}\left({#1}\right)}}
\newcommand{\CDFbT}[1]{\ensuremath{\tilde{F}^{(\mathrm{b},\Theta_t)}_{n}\left({#1}\right)}}
\newcommand{\CDFbTe}[1]{\ensuremath{\hat{F}^{(\mathrm{b},\Theta_t)}_{n}\left({#1}\right)}}

\newcommand{\sumIN}{\sum_{i=1}^{n}}

\newcommand{\abs}[1]{\lvert{#1}\rvert}

\newcommand{\dto}{\xrightarrow{d}}

\newcommand{\expec}[1]{\operatorname{E}\!\left[{#1}\right]}

\newcommand{\texpec}[1]{\operatorname{E}[{#1}]}
\newcommand{\prob}[1]{\operatorname{P}\!\left[{#1}\right]}
\newcommand{\condprob}[1]{\operatorname{P}_\xi\left[{#1}\right]}
\newcommand{\var}[1]{\operatorname{var}\left[{#1}\right]}
\newcommand{\cov}[1]{\operatorname{cov}\left({#1}\right)}

\newcommand{\1}{\bm{1}}

\newcommand{\sign}{\operatorname{sign}}
\newcommand{\ZZ}{\mathbb{Z}}
\newcommand{\RR}{\mathbb{R}}
\newcommand{\NN}{\mathbb{N}}
\newcommand{\reals}{\RR}
\newcommand{\law}[1]{\mathcal{L}\left({#1}\right)}

\newtheorem{theorem}{Theorem}[section]
\newtheorem{lemma}[theorem]{Lemma}

\newtheorem{corollary}[theorem]{Corollary}
\newtheorem{proposition}[theorem]{Proposition}

\theoremstyle{remark}

\newtheorem{remark}[theorem]{Remark}

\numberwithin{equation}{section}
\numberwithin{theorem}{section}

\usepackage[dvipsnames]{xcolor}

\newcommand{\Blue}[1]{\begingroup\color{blue}#1\endgroup}

\allowdisplaybreaks

\makeatletter
\def\namedlabel#1#2{\begingroup
    #2%
    \def\@currentlabel{#2}%
    \phantomsection\label{#1}\endgroup
}
\makeatother

\begin{document}
\setstcolor{red}
\maketitle

\begin{abstract}
To draw inference on serial extremal dependence within heavy-tailed Markov chains, Drees, Segers and Warcho\l{} [Extremes (2015) 18, 369--402] proposed
nonparametric estimators of the spectral tail process. The methodology can be extended to the more general setting of a stationary, regularly varying time series. The large-sample distribution of the estimators is derived via empirical process theory for cluster functionals. The finite-sample performance of these estimators is evaluated via Monte Carlo simulations. Moreover, two different bootstrap schemes are employed which yield confidence intervals for the pre-asymptotic spectral tail process: the stationary bootstrap and the multiplier block bootstrap. The estimators are applied to stock price data to study the persistence of positive and negative shocks.
\end{abstract}

\section{Introduction}
\label{sec:introduction}
The typical modelling paradigm for a time series often starts by choosing a flexible class of models that captures salient features present in the data.  Of course, {\it features} depends on the type of characteristics one is looking for.   For a financial time series consisting of say log-returns of some asset, the key features, often referred to as {\it stylized facts}, include heavy-tailed marginal distributions and serially uncorrelated but dependent data.  These characteristics are readily detected using standard diagnostics such as  qq-plots of the marginal distribution and plots of the sample autocorrelation function (ACF) of the data and the squares of the data.  The GARCH process (and its variants) as well as the  stochastic volatility (SV) process driven by heavy-tailed noise exhibit these attributes and often serve as a starting point for building a model.  More recently, considerable attention has been directed towards studying the extremal behavior of both financial and environmental time series, especially as it relates to estimating risk factors.  Extremes for such time series can occur in clusters and getting a handle on the nature of clusters both in terms of size and frequency of occurrence is important for evaluating various risk measures.  Ultimately, one wants to choose models that adequately describe various extremal dependence features observed in the data.  The theme of this paper is to provide additional tools that not only give  measures of extremal dependence, but can be used as a basis for assessing the quality  of a model's fit to extremal properties present in the data.

The {\it extremal index} $\theta\in (0,1]$ \citep{Leadbetter1983}  is one such measure of extremal dependence for a stationary time series. It is a measure of extremal clustering ($1/\theta$ is the mean cluster size of extremes) with $\theta<1$ indicating clustering and $\theta=1$ signifying no clustering in the limit.  Unfortunately, $\theta$ is a rather crude measure and does not provide fine detail about extremal dependence.  The extremogram, developed in \cite{davis:mikosch:2009}, is an attempt to provide a measure of serial dependence among the extremes in a stationary time series.  It was conceived to be used in much the same way as an ACF in traditional time series modelling, but only applied to extreme values.

In this paper, we will use the spectral tail process, as formulated by \cite{basrak:segers:2009} for heavy-tailed time series, to assess and measure extremal dependence.  The spectral tail process provides a more in-depth description of the structure of extremal dependence than the extremogram.  The first objective of this paper will be to establish limit theory for nonparametric estimates of the distribution of the spectral tail process for a class of heavy-tailed stationary time series.  This builds on earlier work of \cite{DreesSegersWarchol2015} for heavy-tailed Markov chains.  The nonparametric estimates provide quantitative information about extremal dependence within a time series and as such can be used in both exploratory and confirmatory phases of modelling.  As an example, it provides estimates of the probability that an extreme observation will occur at time $t$, given one has occurred at time $0$, and that its absolute value will be even larger. These estimates can also be used for model confirmation, in much the same way that the ACF is used for assessing quality of fit for second-order models of time series.  For example, one can compute a {\it pre-asymptotic} version (to be defined later) of the  distribution  of the spectral tail process from a GARCH process, which in most cases can be easily calculated via simulation.  Then the estimated distribution of the spectral tail process can be compared with the {\it pre-asymptotic} version corresponding to a model for compatibility.  A good fit would indicate the plausibility of using a GARCH model for capturing serial extremal dependence.  The second main objective is then to provide a useful way of measuring compatibility, which we propose using resampling methods.

Recently, there has been increasing interest in the econometric literature for estimating quantities related to extremal dependence. For stochastic processes in continuous time, \cite{bollerslev2013jump} define a $\chi$-coefficient, derived from the extremogram, for assessing tail dependencies applied to financial time series. In a follow-up paper that explores tail risk premia, \cite{bollerslev2015tail} make a connection between their estimates of the time-varying tail shape parameters and the extremogram.
\cite{ linton2007quantilogram} (see also \cite{han2016cross}) introduced the  quantilogram, a diagnostic tool for measuring directional predictability in a time series.   In some respects, our development can be viewed as the {\it quantilogram} for extreme quantiles.  The theory, however, is different in that our quantiles are going to infinity.  Nevertheless, our work does focus on a type of {\it directional predictability}, but only concentrated in the extremes. \cite{ tjostheim2013local} consider local Gaussian correlation and relate it to tail index dependence and the extremogram in a time series context.  Their methodology is applied to financial time series.

The key  object of study in this paper is the tail process and in particular, its normalized version -- the spectral tail process.  A {\it strictly} stationary univariate time series $(X_t)_{t\in \ZZ}$ is said to have a \emph{tail process} $(Y_{t})_{t\in \ZZ}$ if, for all integers $s\le t$, we have 
\begin{equation}
\label{eq:tailprocess}
   \law{ u^{-1} X_s,\ldots,u^{-1} X_t \mid \abs{X_0} > u } \dto \law{ Y_s,\ldots,Y_t }, \qquad u \to \infty,
\end{equation}
with the implicit understanding that the law of $\abs{Y_0}$ is non-degenerate. The law of $\abs{ Y_0 }$ is then necessarily Pareto($\alpha$) for some $\alpha > 0$ and the function $u \mapsto \prob{ \abs{X_0} > u }$ is regularly varying at infinity with index $-\alpha$:
\begin{equation}
\label{eq:RV}
  \lim_{u \to \infty}
  \frac{\prob{ \abs{X_0} > uy }}{\prob{ \abs{X_0} > u }}
  = \prob{ \abs{Y_0} > y }
  = y^{-\alpha}, \qquad y \in [1, \infty).
\end{equation}
The existence of a tail process is equivalent to multivariate regular variation of the finite-dimensional distributions of $(X_t)_{t \in \ZZ}$ \citep[Theorem~2.1]{basrak:segers:2009}.  In many respects, this condition  can be viewed as the heavy-tailed analogue of the condition that a process is Gaussian in the sense that all the finite-dimensional distributions are specified to be of a certain type.

The \emph{spectral tail process} is defined by $\Theta_t = Y_t / \abs{Y_0}$, for $t \in \ZZ$. By \eqref{eq:tailprocess}, it follows that for all integers $s\le t$, we have
\begin{equation}
\label{eq:spectralprocess}
\mathcal{L}\left(X_0/u,X_s / \abs{X_0},\ldots,X_t / \abs{X_0} \mid \abs{X_0} > u \right)
  \dto  \mathcal{L}\left(Y_0, \Theta_s,\dots, \Theta_t \right), \qquad u \to \infty.
\end{equation}
The difference between \eqref{eq:tailprocess} and \eqref{eq:spectralprocess} is that in the latter equation, the variables $X_t$ have been normalized by $\abs{X_0}$ rather than by the threshold $u$. Such auto-normalization allows the tail process to be decomposed into two stochastically independent components, i.e.,
\begin{equation*}
  Y_t = \abs{Y_0} \, \Theta_t, \qquad t \in \ZZ.
\end{equation*}
Independence of $\abs{Y_0}$ and $(\Theta_t)_{t \in \ZZ}$ is stated in \citet[Theorem~3.1]{basrak:segers:2009}. The random variable $\abs{Y_0}$ characterizes the magnitudes of extremes, whereas $(\Theta_t)_{t \in \ZZ}$ captures serial dependence. The spectral tail process at time $t = 0$ yields information on the relative weights of the upper and lower tails of $\abs{X_0}$: since $\Theta_0 = Y_0 / \abs{Y_0} = \sign(Y_0)$, we have
\begin{align}
\label{p_prob}
  p &= \prob{ \Theta_0 = +1 } = \lim_{u \to \infty} \frac{\prob{ X_0 > u }}{\prob{ \abs{X_0} > u }}, &
  1-p &= \prob{ \Theta_0 = -1 }.
\end{align}

The distributions of the \emph{forward} tail process $(Y_t)_{t \ge 0}$ and the \emph{backward} tail process $(Y_t)_{t \le 0}$ mutually determine each other \citep[Theorem~3.1]{basrak:segers:2009}. For all $i,s,t\in\ZZ$ with $s\leq 0\leq t$ and for all measurable functions $f:\reals^{t-s+1} \to \reals$ satisfying $f(y_s,\ldots,y_t)=0$ whenever $y_0=0$, we have, provided the expectations exist,
\begin{equation}
\label{eq:timechange}
  \expec{ f\left(\Theta_{s-i},\ldots,\Theta_{t-i}\right) }
  =
  \expec{
    f \left(
      \frac{\Theta_{s}}{\abs{\Theta_{i}}},
      \ldots,
      \frac{\Theta_{t}}{\abs{\Theta_{i}}}
    \right) \,
    \abs{\Theta_{i}}^\alpha \,
    \1\{ \Theta_i \ne 0 \}
  }.
\end{equation}
The indicator variable $\1\{ \Theta_i \ne 0 \}$ can be omitted because of the presence of $\abs{\Theta_i}^\alpha$, but sometimes, it is useful to mention it explicitly in order to avoid errors arising from division by zero.
By exploiting the `time-change formula' \eqref{eq:timechange}, we will be able to improve upon the efficiency of estimators of the spectral tail process.

Main interest in this paper is in the cumulative distribution function (cdf), $\Ft$, of $\Theta_t$. If $F^{(\Theta_t)}$ is continuous at a point $x$, then
\begin{equation} \label{eq:defcdf}
  \lim_{u \to \infty}
  \prob{ X_t / \abs{ X_0 } \le x \mid \abs{ X_0 } > u }
  =
  \prob{ \Theta_t \le x }
  =
  \Ft(x).
\end{equation}
We consider two  estimates of $\Ft(x)$ based on  forward and backward representations for the tail process.  While these estimates are asymptotically normal, the expressions for the asymptotic variances are too complicated to be useful for constructing confidence regions. To overcome this limitation, inference procedures can be carried out using resampling methods.  Two resampling methods for constructing confidence intervals, based on the stationary bootstrap as used in \cite{Davis2012142}, and the multiplier block bootstrap as described in \cite{drees2015bootstrap}, are applied to our estimates of  $F^{(\Theta_t)}(y)$. In terms of coverage probabilities, the multiplier block bootstrap  performed better than the stationary bootstrap procedure in all the cases we considered.  However, both procedures require care when applied for very high thresholds.

We apply the methodology to study serial
extremal dependence of daily log-returns on the S\&P500 index and the P\&G stock price. We distinguish between two sources of such dependence -- positive and negative shocks -- pointing out an asymmetric behavior. Specifically, we consider cases when extreme values (positive or negative) follow  positive/negative shocks $t$ time lags later.  In terms of the spectral tail process, this corresponds to the probabilities $\prob{\pm \Theta_t>1 \mid \Theta_0=\pm1}$.  We illustrate how well the GARCH(1,1) model and an extension of it allowing for a leverage effect, the APARCH(1,1) model, can capture  serial extremal dependence, as measured by these directional probabilities.   These examples demonstrate how our methodology can provide useful information on the behavior of extremes that follow both positive and negative shocks, which, in turn, can be used in a model-building context.

The remainder of the paper is organized as follows: The two estimates of the tail process are described in Section~\ref{sec:background}, while the companion limit theory for these estimators is formulated in Section~\ref{sec:large_sample_theory}.  The stationary bootstrap and multiplier bootstrap procedures are presented in  Section~\ref{sec:background}.  The validity of the proposed bootstrap methodology is established in Section~\ref{sec:large_sample_theory} too.  The finite-sample performance is investigated through Monte Carlo simulations in Section~\ref{sec:numerical_simulations}.  The application of our methodology is provided in Section~\ref{sec:application}. The proofs of the main results are collected in Section~\ref{sec:appendix}.

\section{Methodology}
\label{sec:background}
\subsection{Estimators}

The data consist of a stretch $X_{1-\tilde{t}}, \ldots, X_{n+\tilde{t}}$, where $\tilde{t}$ is fixed and corresponds to the maximal lag of interest, drawn from a regularly varying, stationary univariate time series with spectral tail process $(\Theta_t)_{t \in \ZZ}$ and index $\alpha > 0$.

In order to estimate $p = \prob{\Theta_0 = 1}$, we simply take the empirical version of \eqref{p_prob}, yielding
\begin{equation*}
  \hat{p}_n
  =
  \frac{\sum_{i=1}^{n} \1\left(X_i > u_n\right)}{\sum_{i=1}^{n}\1\left(\abs{X_i} > u_n\right)}.
\end{equation*}
For $\hat{p}_n$ to be consistent and asymptotically normal, the threshold sequence $u_n$ should tend to infinity at a certain rate described in the next section.

To estimate the cdf, $\Ft$, of $\Theta_t$, we propose the \emph{forward estimator}
\begin{equation}
\label{forward_Tt}
  \CDFfT{x}
  := \frac{\sumIN\1\left(X_{i+t}/\abs{X_i} \leq x,\;\abs{ X_i}>u_n\right)}{\sumIN\1\left( \abs{X_i} > u_n\right)}.
\end{equation}
This is just the empirical version of the left-hand side of \eqref{eq:defcdf}. In equations~\eqref{eq:tailprocess} and \eqref{eq:spectralprocess}, the conditioning event is $\{ \abs{X_0} > u \}$, making no distinction between positive extremes, $X_0 > u$, and negative extremes, $X_0 < -u$. However, these two cases can be distinguished by conditioning on the sign of $\Theta_0$. In particular, we define
\begin{equation}
\label{forward_ABt}
 \hat{F}^{(\mathrm{f},\Theta_t \mid \Theta_0=\pm1)}_{n}\left(x\right)
:= \frac{\sumIN\1\left(X_{i+t}/\left(\pm X_i\right) \leq x,\; \pm X_i > u_n\right)}{\sumIN\1\left( \pm X_i >  u_n\right)}.
\end{equation}

The numerator in the estimator is a sum of indicator functions, most of which are zero. This often leads to a large variance. The time-change formula~\eqref{eq:timechange} yields a different representation of the law of $\Theta_t$, motivating a different estimator than the one above. Depending on the value of $x$, the new estimator will involve more non-zero indicators, which receive weights instead. The simulation study reported in Section~\ref{sec:sim:estimators} will show that the resulting estimator may have a smaller variance than the one in \eqref{forward_ABt}, in particular if $|x|$ is large.

\begin{lemma}
\label{lem:tc:T}
Let $(X_t)_{t \in \ZZ}$ be a stationary univariate time series, regularly varying with index $\alpha$ and spectral tail process $(\Theta_t)_{t \in \ZZ}$. Then, for all integer $t \neq 0$,
 \begin{equation}\label{eq:Tt:tc}
 \prob{ \Theta_t \le x }=
\begin{cases}
1-\expec{
    \abs{\Theta_{-t}}^{\alpha} \,
    \1\left(\Theta_0/\abs{\Theta_{-t}} > x\right)
  }&
\text{if $x \geq 0$}, \\[1em]
  \expec{
    \abs{\Theta_{-t}}^{\alpha}\,
   \1\left(\Theta_0/\abs{\Theta_{-t}} \leq x\right)
  }&
\text{if $x<0$.}
\end{cases}
 \end{equation}
 Moreover
 \begin{equation}\label{eq:At:tc}
 \prob{ \Theta_t \le x \mid \Theta_0=1}=
\begin{cases}
1-\frac{1}{p}\expec{
    \Theta_{-t}^{\alpha} \,
    \1\left(1/\Theta_{-t} > x,\;\Theta_0=1\right)
  }&
\text{if $x \geq 0$}, \\[1em]
  \frac{1}{p}\expec{
    \Theta_{-t}^{\alpha}\,
   \1\left(-1/\Theta_{-t}\leq x,\;\Theta_0=-1\right)
  }&
\text{if $x<0$,}
\end{cases}
 \end{equation}
 and
  \begin{equation}\label{eq:Bt:tc}
 \prob{ \Theta_t \le x \mid \Theta_0=-1}=
\begin{cases}
1-\frac{1}{1-p}\expec{
    \left(-\Theta_{-t}\right)^{\alpha} \,
    \1\left(-1/\Theta_{-t} > x,\;\Theta_0=1\right)
  }&
\text{if $x \geq 0$}, \\[1em]
  \frac{1}{1-p}\expec{
    \left(-\Theta_{-t}\right)^{\alpha}\,
   \1\left(1/\Theta_{-t} \leq x,\;\Theta_0=-1\right)
  }&
\text{if $x<0$.}
\end{cases}
 \end{equation}
\end{lemma}


If population quantities are replaced by their sample counterparts, Lemma~\ref{lem:tc:T} suggests the following \emph{backward estimator} of the cdf of $\Theta_t$:
\begin{equation}
\label{backward_Tte}
  \CDFbTe{x}
:=
\begin{cases}
1-\dfrac{\sumIN\abs{ \frac{ X_{i-t}}{X_i}}^{\hat\alpha_n}\1\left(X_i/ \abs{X_{i-t}} >x,
\; \abs{X_i} > u_n \right)}{\sumIN\1\left( \abs{ X_i} > u_n\right)} &
\text{if $x \geq 0$}, \\[1em]
\dfrac{\sumIN\abs{ \frac{X_{i-t}}{X_i}}^{\hat\alpha_n}\1\left(X_i/ \abs{X_{i-t}} \leq x,
\; \abs{X_i} > u_n \right)}{\sumIN\1\left(  \abs{X_i} > u_n\right)} &
\text{if $x<0$.}
\end{cases}
\end{equation}
Here, $\hat{\alpha}_n$ is an estimator of the tail index, for which we will take the Hill-type estimator
\begin{equation}
\label{Hill}
  \hat{\alpha}_n = \frac{\sum_{i=1}^{n} \1\left( \abs{ X_i } > u_n\right)}{\sum_{i=1}^{n} \log \left(\abs{ X_{i} }/ u_n\right) \1\left(\abs{ X_i } > u_n\right)}.
\end{equation}
Conditioning on an extreme value of a specific sign, we get
\begin{equation*}
\hat{F}^{(\mathrm{b},\Theta_t \mid \Theta_0=\pm1)}_{n}\left(x\right)
:=
\begin{cases}
1-\dfrac{\sumIN\left(\frac{ \pm X_{i-t}}{X_i}\right)^{\hat{\alpha}_n}\1\left(\pm X_i/ X_{i-t} >x,\;
\; X_i> u_n \right)}{\sumIN\1\left( \pm X_i > u_n\right)} &
\text{if $x \geq 0$}, \\[1em]
\dfrac{\sumIN\left(\frac{\mp X_{i-t}}{X_i}\right)^{\hat{\alpha}_n}\1\left(\pm X_i/ X_{i-t} \leq x,\;
\; X_i < -u_n \right)}{\sumIN\1\left(  \pm X_i  > u_n\right)} &
\text{if $x<0$.}
\end{cases}
\end{equation*}
The asymptotic and finite-sample distributions of these estimators will be investigated in the following sections.

\subsection{Resampling}
\label{sec:resampling}

We explore two different bootstrap schemes that yield confidence intervals for $\Ft(x)$, or rather, for the pre-asymptotic version $\prob{ X_t / \abs{X_0} \le x \mid \abs{X_0} > u }$: the stationary bootstrap and the multiplier block bootstrap. We apply each of the two resampling schemes to both the forward and backward estimators at various levels $x$ and at different lags $t$.

The stationary bootstrap goes back to \citet{politis1994stationary} and is an adaptation of the block bootstrap by allowing for random block sizes. The resampling scheme was applied to the extremogram in \citet{Davis2012142}. It consists of generating
pseudo-samples $X^*_1,\ldots,X^*_n$, drawn from the sample $X_1,\ldots,X_n$ by taking the first $n$ values in the sequence
\[
  X_{K_1},\ldots,X_{K_1 + L_1-1},X_{K_2},\ldots,X_{K_2+L_2-1},\ldots,
\]
where $K_1,K_2\ldots$ is an iid sequence of random variables uniformly distributed on $\{1,\ldots,n\}$ and $L_1,L_2,\ldots$ is an iid sequence of geometrically distributed random variables (independent of $(K_j)_{j\in\NN}$) with distribution $\prob{L_1=l}=p\left(1-p\right)^{l-1}$, $l=1,2,\ldots$ for some $p=p_n\in(0,1)$ such that $p_n\rightarrow 0$ and $np_n\rightarrow \infty$. If the index $t$ thus obtained exceeds the sample size $n$, we replace $t$ by $(t-1 \, \operatorname{mod} \, n)+1$, i.e., we continue from the beginning of the sample. The estimators are then applied to $X_{1-\tilde t}^*,\ldots,X_{n+\tilde t}^*$.

The multiplier block bootstrap method was applied to cluster functionals in \cite{drees2015bootstrap}. It consists of splitting the data set into $m_n = \lfloor n/r_n \rfloor$ blocks of length $r_n$ and multiplying the cluster functionals of each block by a random factor. (Here $\lfloor x\rfloor$ denotes the integer part of $x$.) Specifically, for iid random variables $\xi_j$, independent of $(X_t)_{t \in \ZZ}$, with $\expec{\xi_j}=0$ and $\var{\xi_j}=1$, the bootstrapped forward estimator can be written as
\begin{equation*}
  \hat{F}^{*(\mathrm{f},\Theta_t)}_{n}\left(x\right)
  :=\frac{\sum_{j=1}^{m_n} (1+\xi_j) \sum_{i\in I_j}\1\left(\frac{X_{i+t}}{\abs{X_{i}} }\leq x,\;\abs{ X_{i}}>u_n\right)}{\sum_{j=1}^{m_n} (1+\xi_j) \sum_ {i\in I_j}\1\left( \abs{X_{i}} > u_n\right)},
\end{equation*}
where $I_j = \{(j-1)r_n+1,\ldots,jr_n\}$ denotes the set of indices belonging to the $j$th block. Similarly, the bootstrapped backward estimator for $x > 0$ with estimated index of regular variation is
\begin{equation}
\label{eq:defbootbackward}
  \hat{F}^{*({\mathrm{b}},\Theta_t)}_{n}\left(x\right)
  := 1-\dfrac{\sum_{j=1}^{m_n} (1+\xi_j) \sum_ {i\in I_j}\abs{ \frac{ X_{i-t}}{X_{i}}}^{\hat\alpha^*_n}\1\left(\frac{X_{i}}{\abs{X_{i-t}}} >x,
\; \abs{X_{i}} > u_n \right)}{\sum_{j=1}^{m_n}(1+\xi_j)\sum_ {i\in I_j}\1\left( \abs{X_{i}} > u_n\right)}
\end{equation}
with
\begin{equation}
\label{eq:Hill:boot}
 \hat\alpha_n^* := \frac{\sum_{j=1}^{m_n} (1+\xi_j) \sum_{i\in I_j} \1( \abs{X_{i}} > u_n)}{\sum_{j=1}^{m_n}(1+\xi_j)\sum_ {i\in I_j} \log(\abs{X_i}/u_n) \, \1( \abs{X_{i}} > u_n)}.
\end{equation}

If the threshold $u_n$ is high, it may be advisable to construct bootstrap confidence intervals based on lower thresholds and then scale accordingly; see the explanation after Theorem~\ref{theo:bootstrap}.

\subsection{Testing for dependence of extreme observations}

For iid random variables, the spectral tail process simplifies to $\Theta_t \equiv 0$ a.s.\ for all nonzero $t$. If this occurs for a stationary, regularly varying time series, then we say that the series exhibits \emph{serial extremal independence}. The opposite case is referred to as \emph{serial extremal dependence}, i.e., at least one of the variables $\Theta_t$ for $t \neq 0$ is not degenerate at $0$. Since the convergence of the pre-asymptotic distribution can be arbitrarily slow, one cannot formally test for extremal dependence  within the present framework.

However, if one wants to test whether the exceedances over a given high threshold $u$ are independent then one may check whether the lower bound of a confidence interval for, say, $\prob{ \abs{X_t} \ge \abs{X_0} \mid \abs{X_0} > u }$ constructed by one of the bootstrap methodologies is larger than this probability under the assumption of exact independence of $\abs{X_0}\1(\abs{X_0}>u)$ and $\abs{X_t}\1(\abs{X_t}>u)$, which is easily shown to equal $\prob{ \abs{X_0} > u } / 2$.

If one prefers to work with exceedances of the original time series (instead of its absolute values), then the probability under the assumption of independence depends on the relative weights of the upper and lower tails, and can thus not be calculated analytically. In that case, it seems natural to calculate this probability by Monte Carlo simulation. To this end, one generates (conditionally) iid samples according to the empirical distribution of the original time series by drawing with replacement from the observations, which corresponds to the classical bootstrap procedure for iid data. The considered probability under independence can be approximated by the pertaining relative frequency, which is then compared with the lower confidence bound for the probability  estimated from the original time series. If the latter is larger this indicates that the exceedances in the time series exhibit a non-negligible serial dependence (see Figure \ref{plot:app1} in Section \ref{sec:application}).

Alternatively, one may compare the  pre-asymptotic probability estimated from the observed time series using either the forward or the backward estimator with quantiles of the distribution of this estimator under independence. In a similar way as described above, the latter can be approximated by an empirical quantile obtained in Monte Carlo simulations with (conditionally) iid samples (cf.\ Figure \ref{plot:app3}).

\section{Large-sample theory}
\label{sec:large_sample_theory}
Under certain conditions, the standardized estimation errors of the forward and the backward estimators converge jointly to a centered Gaussian process (Section~\ref{sec:CLT:estimators}). Convergence of the multiplier block bootstrap follows under the same conditions (Section~\ref{sec:CLT:bootstrap}).

In order not to overload the presentation, we focus on nonnegative time series. We briefly indicate how the conditions and results must be modified in the real-valued case.

\subsection{Asymptotic normality of the estimators}
\label{sec:CLT:estimators}

 All estimators under consideration can then be expressed in terms of \emph{generalized tail array sums}. These are statistics of the form $\sumIN\phi(X_{n,i})$, with
\begin{equation}
\label{eq:Xni}
  X_{n,i}
  :=
  u_n^{-1}
  \left(X_{i-\tilde{t}},\ldots,X_i,\ldots,X_{i+\tilde{t}}\right)
  \1(X_i>u_n).
\end{equation}
\cite{drees2010limit} give conditions under which, after standardization, such statistics converge to a centered Gaussian process, uniformly over appropriate families of functions $\phi$. From these results we will deduce a functional central limit theorem for the processes of forward and backward estimators defined in \eqref{forward_Tt} and \eqref{backward_Tte} with $\hat\alpha_n$ according to \eqref{Hill}, respectively.

To ensure consistency, the threshold $u_n$ must tend to infinity in such a way that
\[
  v_n:=\prob{X_0>u_n}
\]
tends to $0$, but the expected number, $nv_n$, of exceedances tends to infinity.
Moreover, we have to ensure that observations which are sufficiently separated in time are almost independent. The strength of dependence will be assessed by the $\beta$-mixing coefficients
\[
  \beta_{n,k} := \sup_{1\le l\le n-k-1}
  \expec{\sup_{B\in\mathcal{B}_{n,l+k+1}^n} \left\lvert\prob{B \mid \mathcal{B}_{n,1}^l}-\prob{B}\right\rvert}.
\]
Here $\mathcal{B}_{n,i}^j$ is the $\sigma$-field generated by $(X_{n,l})_{i\le l\le j}$.

We assume that there exist sequences $l_n,r_n\to\infty$ and some $x_0\ge 0$ such that the following conditions hold:
\begin{description}
\item[\namedlabel{con:A}{(A($x_0$))}]
The cdf of $\Theta_t$, $F^{(\Theta_t)}$, is continuous on $[x_0,\infty)$, for $t\in\{1,\ldots,\tilde{t}\}$.
\item[\namedlabel{con:B}{(B)}]
As $n \to \infty$, we have $l_n\to\infty$, $l_n=o(r_n)$, $r_n=o((n v_n)^{1/2})$, $r_nv_n\to 0$, and $\beta_{n,l_n}n/r_n\to 0$.

\item[\namedlabel{con:C'}{(C)}]
For all $k \in \{0, \ldots, r_n\}$, there exists
\begin{equation}
 \label{eq:snkdef2}
s_n(k)\ge \expec{\log\Big(\frac{X_0}{u_n}\Big) \max\Big\{\log\Big(\frac{X_k}{u_n}\Big),\1(X_k>u_n)\Big\} \, \Big| \, X_0>u_n}
\end{equation}
such that $s_\infty(k)=\lim_{n\to\infty}s_n(k)$ exists, $\lim_{n\to\infty} \sum_{k=1}^{r_n} s_n(k) = \sum_{k=1}^\infty s_\infty(k)$ holds and the last sum is finite.

Moreover,  there exists $\delta>0$ such that
\begin{equation}
 \label{eq:psibdd}
  \sum_{k=1}^{r_n}
  \bigg(
    \expec{
      \Bigl( \log^+ \Bigl( \frac{X_0}{u_n}\Big)\log^+\Big(\frac{X_k}{u_n} \Bigr) \Bigr)^{1+\delta}
      \, \Big| \,
      X_0 > u_n
    }
  \bigg)^{1/(1+\delta)}
  = O(1),
  \qquad n \to \infty.
\end{equation}
\end{description}
Without Condition~\ref{con:A} one cannot expect uniform convergence of the estimated cdf of $\Theta_t$ to the true cdf on $[x_0,\infty)$. Indeed, in this case even $\prob{X_t/X_0\le x \mid |X_0|>u}$ need not converge to $F^{(\Theta_t)}(x)$ for a point of discontinuity $x$.
Condition~\ref{con:B} imposes restrictions on the rate at which $v_n$ tends to 0 and thus on the rate at which $u_n$ tends to $\infty$. Often, the $\beta$-mixing coefficients decay geometrically, i.e., $\beta_{n,k}=O(\eta^k)$ for some $\eta\in(0,1)$. Then one may choose $l_n=O(\log n)$, and Condition~\ref{con:B} is fulfilled for a suitably chosen $r_n$ if $(\log n)^2/n=o(v_n)$ and $v_n =o(1/(\log n))$.

The technical Condition~\ref{con:C'} rules out too large a cluster of extreme observations. Using integration by parts, the right-hand side of \eqref{eq:snkdef2} can be bounded by
$$ v_n^{-1} \int_1^\infty \left(\prob{X_0>u_n s,X_k>u_n}+\int_1^\infty \prob{X_0>u_n s,X_k>u_nt} t^{-1}\, dt\right) s^{-1} ds.
$$
Now one can use techniques employed in  \cite{drees2000tailquantile}   and \cite{drees2003quantile} to verify \eqref{eq:snkdef2} for specific time series models like solutions to stochastic recurrence equations or suitable heavy tailed linear time series. (Typically, the upper bounds $s_n(k)$ are of the form $\rho_k+\xi_n$ for a summable sequence $\rho_k$ and $\xi_n=o(1/r_n)$.) The left-hand side of \eqref{eq:psibdd} can be rewritten in the form
$$ \sum_{k=1}^{r_n} \left( (1+\delta)^2 v_n^{-1} \int_1^\infty \int_1^\infty \prob{X_0>u_n s,X_k>u_nt} (\log s\log t)^\delta (st)^{-1}\, ds\, dt\right)^{1/(1+\delta)},
$$
which can then be bounded by similar techniques.

Under these conditions, one can prove the asymptotic normality of relevant generalized tail array sums (see Proposition \ref{prop:procconv} below) and thus the joint uniform asymptotic normality of the appropriately centered forward and the backward estimator of $ F^{(\Theta_t)}$.
\begin{theorem}
\label{theo:asnormest}
Let $(X_t)_{t\in\ZZ}$ be a stationary, regularly varying process. If \ref{con:A}, \ref{con:B} and \ref{con:C'}  are fulfilled for some $x_0\ge 0$ and $y_0\in [x_0,\infty)\cap (0,\infty)$, then
\begin{multline} \label{eq:weakconv}
  (nv_n)^{1/2}
  \begin{pmatrix}
    \big( \CDFfT{x_t}- \prob{X_t/X_0\le x_t\mid X_0>u_n}\big)_{x_t\in[x_0,\infty)} \\
    \big( \CDFbTe{y_t}- (1-\expec{(X_{-t}/X_0)^\alpha \1(X_0/X_{-t}>y_t)\mid X_0>u_n})\big)_{y_t\in[y_0,\infty)}
  \end{pmatrix}_{|t|\in\{1,\ldots,\tilde{t}\}}
  \dto \\
   \begin{pmatrix}
    ( Z(\phi_{2,x_t}^t)-\bar F^{(\Theta_t)}(x_t) Z(\phi_1))_{x_t\in[x_0,\infty)} \\
    (Z(\phi_{3,x_t}^t)-\bar F^{(\Theta_t)}(y_t) Z(\phi_1)+(\alpha^2 Z(\phi_0)-\alpha Z(\phi_1))\expec{ \log (\Theta_t) \, \1(\Theta_t>y_t)})_{y_t\in[y_0,\infty)}
  \end{pmatrix}_{|t|\in\{1,\ldots,\tilde{t}\}}
\end{multline}
where $Z$ is a centered Gaussian process, indexed by functions defined in \eqref{eq:functions}, whose covariance function is given in \eqref{eq:cov_emp_pr}, and $\bar F^{(\Theta_t)}:=1-F^{(\Theta_t)}$ denotes the survival function of $\Theta_t$. (Assertion \eqref{eq:weakconv} means that for suitable versions of the processes the convergence holds uniformly for all $x_t\ge x_0, y_t\ge y_0$ and $|t|\in\{1,\ldots,\tilde{t}\}$ almost surely.)
\end{theorem}

 Additional conditions are needed to ensure that the biases of the forward and the backward estimator of $F^{(\Theta_t)}$ are asymptotically negligible:
\begin{eqnarray}
  \sup_{x\in [x_0,\infty)} \left\lvert \prob {\frac{X_t}{X_0} \le x\,\Big|\, X_0>u_n} -  F^{(\Theta_t)}(x) \right\rvert
  & = & o\bigl((nv_n)^{-1/2}\bigr),  \label{eq:bias1}\\
  \sup_{y\in [y_0,\infty)} \left\lvert \expec{\Big(\frac{X_{-t}}{X_0}\Big)^\alpha \1(X_0/X_{-t}>y)\,\Big|\, X_0>u_n} - \bar F^{(\Theta_t)}(y) \right\rvert
  & = & o\bigl((nv_n)^{-1/2}\bigr),\label{eq:bias2}\\
  \bigl\lvert
    \texpec{\log(X_0/u_n)\mid X_0>u_n}- 1/\alpha
  \bigr\rvert
  & = & o \bigl( (nv_n)^{-1/2} \bigr), \label{eq:bias3}
\end{eqnarray}
for $t \in \{-\tilde{t}, \ldots, \tilde{t}\} \setminus \{0\}$ as $n\to\infty$.  These conditions are fulfilled if $nv_n$ tends to $\infty$ sufficiently slowly, because by definition of the spectral tail process, the regular variation of $X_0$ and by \eqref{eq:Tt:tc}, the left-hand sides in \eqref{eq:bias1}--\eqref{eq:bias3} tend to $0$ if $F^{(\Theta_t)}$ is continuous on $[x_0,\infty)$.

\begin{corollary}
\label{corol:probcenter}
Let $(X_t)_{t\in\ZZ}$ be a stationary, regularly varying process. If \ref{con:A}, \ref{con:B}, \ref{con:C'}, and \eqref{eq:bias1}--\eqref{eq:bias3}  are fulfilled for some $x_0\ge 0$ and $y_0\in [x_0,\infty)\cap (0,\infty)$, then
\begin{multline*}
  (nv_n)^{1/2}
  \begin{pmatrix}
    ( \CDFfT{x_t}- F^{(\Theta_t)}(x_t))_{x_t\in[x_0,\infty)} \\
    ( \CDFbTe{y_t}- F^{(\Theta_t)}(y_t))_{y_t\in[y_0,\infty)}
  \end{pmatrix}_{|t|\in\{1,\ldots,\tilde{t}\}}
  \dto \\
   \begin{pmatrix}
    ( Z(\phi_{2,x_t}^t)-\bar F^{(\Theta_t)}(x_t) Z(\phi_1))_{x_t\in[x_0,\infty)} \\
    (Z(\phi_{3,x_t}^t)-\bar F^{(\Theta_t)}(y_t) Z(\phi_1)+(\alpha^2 Z(\phi_0)-\alpha Z(\phi_1))\expec{ \log (\Theta_t) \, \1(\Theta_t>y_t)})_{y_t\in[y_0,\infty)}
  \end{pmatrix}_{|t|\in\{1,\ldots,\tilde{t}\}}
\end{multline*}
where $ Z$ is the centered Gaussian process defined in Theorem \ref{theo:asnormest}.
\end{corollary}

In general, it is difficult to compare the asymptotic variances of the backward and the forward estimator.

\subsection{Consistency of the multiplier block bootstrap}
\label{sec:CLT:bootstrap}

Here we discuss the asymptotic behavior of the multiplier block bootstrap version of the forward and backward estimators. For the sake of brevity, we focus on  estimators of $F^{(\Theta_t)}(x)$ for a fixed $x$.

\cite{drees2015bootstrap} has shown convergence of bootstrap versions of empirical processes of tail array sums under the same conditions needed for convergence of the original empirical processes. Let $\operatorname{P}_\xi$ denote the probability w.r.t.\ $\xi=(\xi_j)_{j\in\NN}$, i.e., the conditional probability given $(X_{n,i})_{1\le i\le n}$.
\begin{theorem}
\label{theo:bootstrap}
Let $\xi_j$, $j\in\NN$, be iid random variables independent of $(X_t)_{t\in\ZZ}$ with $\expec{\xi_j}=0$ and $\var{\xi_j}=1$. Then, under the conditions of Theorem \ref{theo:asnormest}, for all $x\ge x_0, y\ge y_0$,
\begin{eqnarray*}
 \lefteqn{\sup_{r,s\in \mathbb{R}^{2\tilde{t}}}
 \bigg|\operatorname{P}_\xi\Big[(nv_n)^{1/2}\left(\hat{F}^{*(\mathrm{f},\Theta_t)}_{n}(x)- \CDFfT{x}\right)\le r_t, }
 \\
 & & \hspace*{3cm} (nv_n)^{1/2}\left(\hat{F}^{*({\mathrm{b}},\Theta_t)}_{n}(y)- \CDFbTe{y}\right)\le s_t,\;\forall\,|t|\in\{1,\ldots,\tilde{t}\}\Big] 
 \\
 & &\hspace*{0.3cm} - \operatorname{P}\Big[(nv_n)^{1/2}\left(\hat{F}^{(\mathrm{f},\Theta_t)}_{n}(x)- F^{(\Theta_t)}(x)\right)\le r_t, 
 \\
  & & \hspace*{3cm}(nv_n)^{1/2}\left(\CDFbTe{y}- F^{(\Theta_t)}(y)\right)\le s_t,\;\forall\,|t|\in\{1,\ldots,\tilde{t}\}\Big]
  \bigg| \to 0
\end{eqnarray*}
in probability.
\end{theorem}

In particular, if $a$ and $b$ are such that $\condprob{\hat{F}^{*({\mathrm{b}},\Theta_t)}_{n}(y)\in [a,b]}=\beta$, then
\[
  \left[ 2\CDFbTe{y}-b, \; 2\CDFbTe{y}-a \right]
\]
is a confidence interval for $F^{(\Theta_t)}(y)$ with approximative coverage level $\beta$. However, if the number of exceedances over a given threshold is too small, one may prefer to construct confidence intervals based on bootstrap estimators corresponding to lower thresholds. Let $\tilde{u}_n$ denote another threshold sequence, let $\tilde{v}_n = \prob{X_0 > \tilde{u}_n}$ denote the corresponding exceedance probabilities, and let $\hat{\tilde F}^{({\mathrm{b}},\Theta_t)}_{n}(y)$ and $\hat{\tilde F}^{*({\mathrm{b}},\Theta_t)}_{n}(y)$ denote the backward estimator and the bootstrap version thereof, respectively, based on the exceedances over $\tilde u_n$. The conditional distribution of
\[
  (n\tilde v_n)^{1/2}
  \left(
    \hat{\tilde F}^{*({\mathrm{b}},\Theta_t)}_{n}(y)
    -
    \hat{\tilde F}^{({\mathrm{b}},\Theta_t)}_{n}(y)
  \right)
\]
given the data is approximately the same as the unconditional distribution of
\[
  (nv_n)^{1/2} \left( \CDFbTe{y}-F^{(\Theta_t)}(y) \right).
\]
Hence, if $a$ and $b$ are such that $\condprob{\hat{\tilde F}^{*({\mathrm{b}},\Theta_t)}_{n}(y)\in [a,b]}=\beta$ and if $\hat{\tilde{v}}_n/\hat{v}_n$ is a suitable estimator of $\tilde v_n/v_n$, then
\begin{equation}  \label{eq:rescalebootconf}
  \left[
    \left(\frac{\hat{\tilde{v}}_n}{\hat{v}_n}\right)^{1/2}
    \big(\hat{\tilde F}^{({\mathrm{b}},\Theta_t)}_{n}(y)-b\big)+\CDFbTe{y}, \;
    \left(\frac{\hat{\tilde{v}}_n}{\hat{v}_n}\right)^{1/2}
    \big(\hat{\tilde F}^{({\mathrm{b}},\Theta_t)}_{n}(y)-a\big)+\CDFbTe{y}
  \right]
\end{equation}
is a confidence interval  for $F^{(\Theta_t)}(y)$
with approximative coverage probability $\beta$. In practice, one will often use large order statistics as thresholds, say the $k_n$-th and $\tilde{k}_n$-th largest observations, respectively. In that case, $\tilde v_n/v_n$ can be replaced by $\tilde k_n/k_n$. A similar approach, namely to use a variance estimator which is based on a lower threshold, has successfully been employed in \citet[Section~5]{drees2003quantile}.

Of course, confidence intervals based on the bootstrap version of the forward estimator can be constructed analogously.

\begin{remark}
It is possible to generalize Theorem~\ref{theo:bootstrap} to cover the joint limit distribution of the bootstrap estimators for all $x \ge x_0$ and $y \ge y_0$. Technically, this requires to endow the space of probability measures on spaces of bounded functions from $[x_0,\infty)$ (resp.\ $[y_0,\infty)$) to $\mathbb{R}^{2\tilde{t}}$ with a metric that induces weak convergence, e.g., the bounded Lipschitz metric. This is the approach in \cite{drees2015bootstrap}  to establish the consistency of a bootstrap method for estimating the extremogram, a close cousin of the tail process. Based on such a result, one may construct uniform confidence bands for the function $F^{(\Theta_t)}$ on $[x_0,\infty)$ or $[y_0,\infty)$, respectively, which in general will be considerably wider than the pointwise confidence intervals discussed above and will thus often be rather uninformative. For brevity, we omit the details.
\end{remark}

\begin{remark}
For time series which may take on negative values too, the forward and backward estimators of $ F^{(\Theta_t)}$  can be represented in terms of generalized tail array sums constructed from
\[
X_{n,i}^{\tilde{t}}=u_n^{-1}\left(X_{i-\tilde{t}},\ldots,X_i,\ldots,X_{i+\tilde{t}}\right)
\1\left(\abs{X_i}>u_n\right).
\]
When $x < 0$, for example, the backward estimator $\CDFbT{x}$ is equal to the ratio of the generalized tail array sums pertaining to the functions
\begin{align*}
  (y_{-\tilde{t}},\ldots,y_0,\ldots,y_{\tilde{t}}) & \mapsto \abs{y_{-t}/y_0}^\alpha \, \1(y_0/\abs{y_{-t}}\le x, \abs{y_0}>1), \\
 (y_{-\tilde{t}},\ldots,y_0,\ldots,y_{\tilde{t}}) &\mapsto \1(\abs{y_0}> 1).
\end{align*}
Limit theorems can be obtained by the same methods as in the case of non-negative observations under obvious analogues to the conditions \ref{con:A}, \ref{con:B} and \ref{con:C'} with $v_n:=\prob{\abs{X_0}>u_n}$.
\end{remark}

\section{Finite-sample performance}
\label{sec:numerical_simulations}
In Section~\ref{sec:sim:estimators}, we show results from a numerical simulation study designed to test the performance of the forward \eqref{forward_Tt}  and the backward \eqref{backward_Tte} estimators. We continue in  Section~\ref{sec:sim:bootstrap} by evaluating the performance of two bootstrap schemes, the multiplier block bootstrap and the stationary bootstrap, described in \Blue{Section~\ref{sec:resampling}.}

The simulations are based on pseudo-random samples from two widely used models for financial time series. Both models are of the form $X_t = \sigma_t Z_t$ where $\sigma_t$ and $Z_t$ are independent. First, we consider the GARCH(1,1) model with $\sigma_t^2=0.1+ 0.14X_{t-1}^2+ 0.84\sigma_{t-1}^2$, the innovations $Z_t$ being independent $t_4$ random variables, standardized to have unit variance. The second model is the stochastic volatility (SV) process with $\log{\sigma_t}=0.9\log{\sigma_{t-1}} +\epsilon_t$, with independent standard normal innovations $\epsilon_t$ and independent innovations $Z_t$ with common distribution $t_{2.6}$. The parameters have been chosen to ensure that both time series are regularly varying with index $\alpha = 2.6$ \citep{davis2001,mikosch2000}.

\subsection{Forward and backward estimators}
\label{sec:sim:estimators}

We estimate $\prob{\Theta_t \leq x}$ for both the GARCH(1,1) and the SV model, for various arguments $x$ and lags~$t$, via the forward and the backward estimator, with estimated tail index $\alpha$. The threshold is set at the empirical $95\%$ quantile of the absolute values of a time series of length $n = 2\,000$. We do $1\,000$ Monte Carlo repetitions and calculate bias, standard deviation, and root mean squared error (RMSE) with respect to the pre-asymptotic values $\prob{X_t/|X_0|\le x\mid |X_0|>F_{|X|}^\leftarrow(0.95)}$ in the forward representation. The true quantile $F_{|X|}^\leftarrow(0.95)$ of $|X_0|$ and the true pre-asymptotic values were calculated numerically via $10\,000$ Monte Carlo simulations based on time series of length $10\,000$.

It was already reported in the context of Markovian time series that for $t=1$ and $\abs{x}$ large, the backward estimators usually have a smaller variance than the forward estimators \citep{DreesSegersWarchol2015}. Here, numerical simulations suggest that this is true for non-Markovian time series and for higher lags as well. The results are presented in Figure~\ref{plot:for_back_rmse}.

\begin{figure}[htp]
\begin{center}
\begin{tabular}{ccc}
\includegraphics[width=0.3\textwidth]{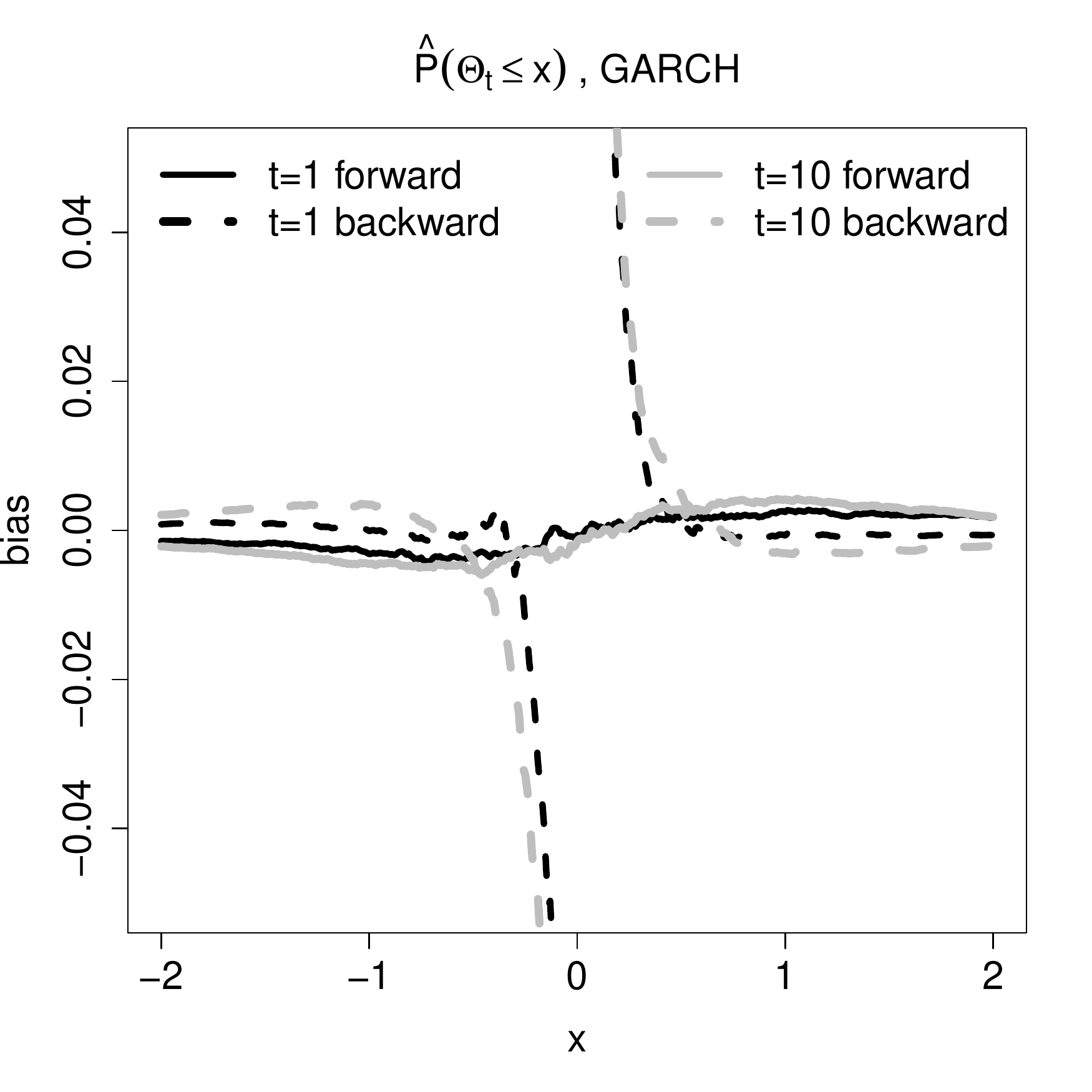}
&
\includegraphics[width=0.3\textwidth]{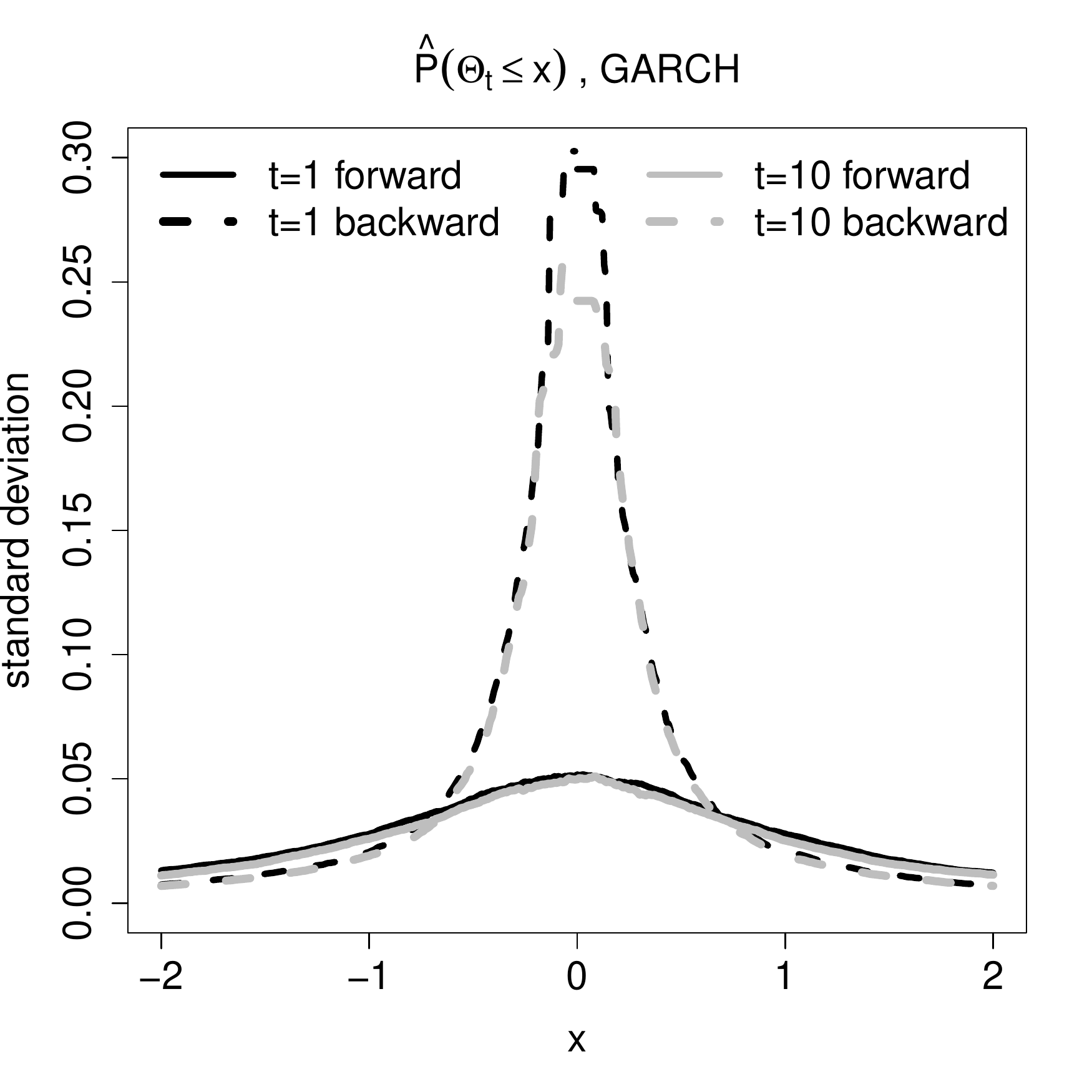}
&
\includegraphics[width=0.3\textwidth]{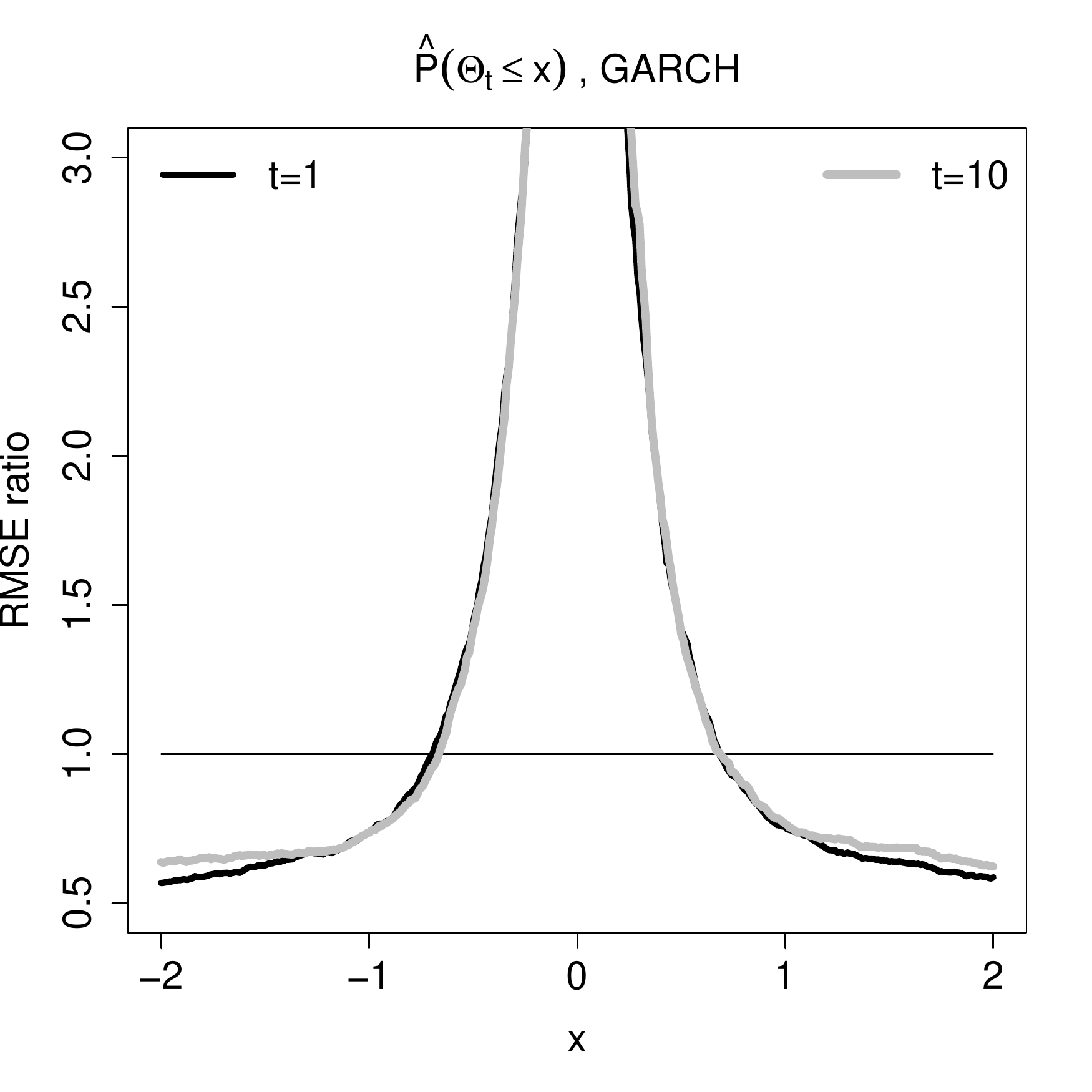}\\

\includegraphics[width=0.3\textwidth]{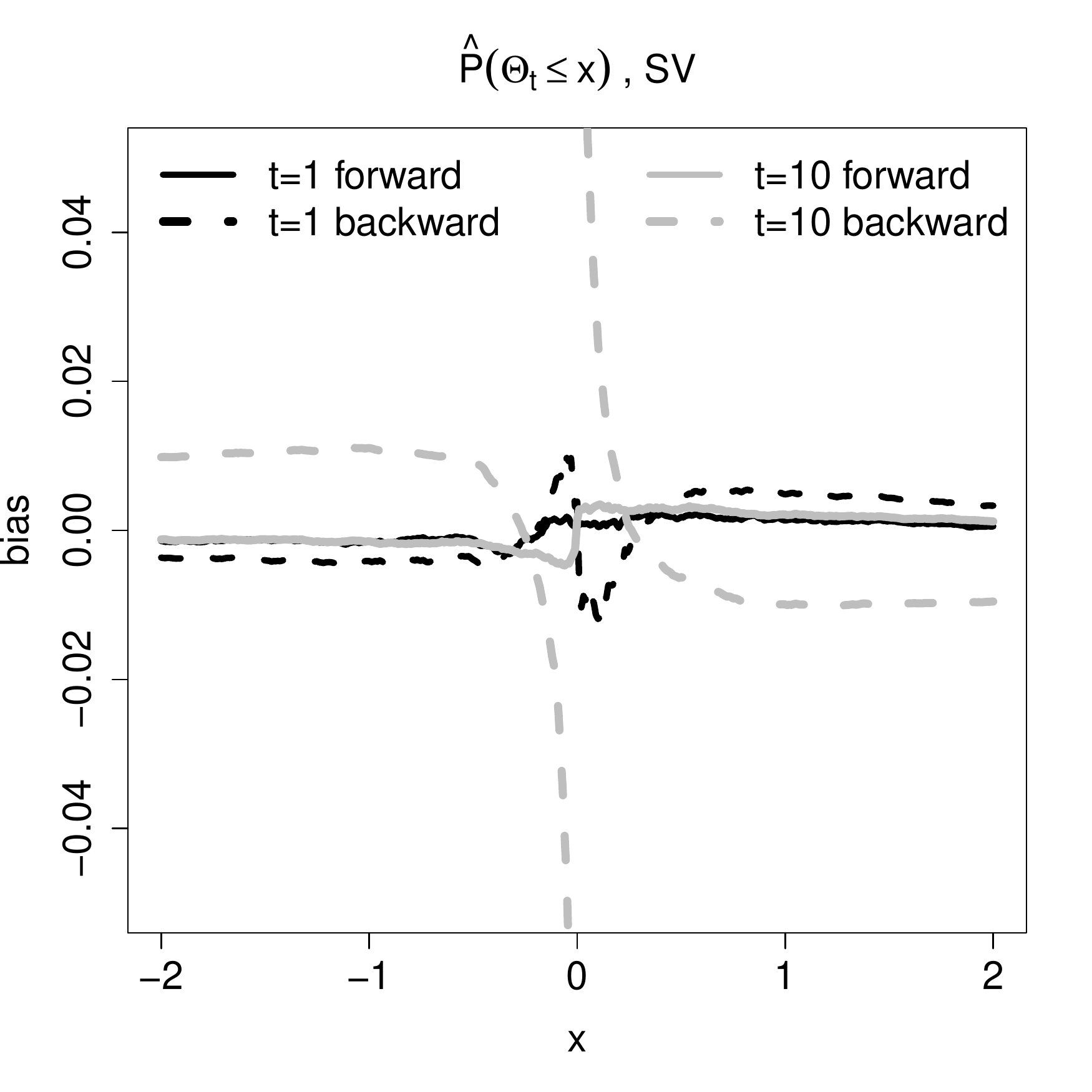}
&
\includegraphics[width=0.3\textwidth]{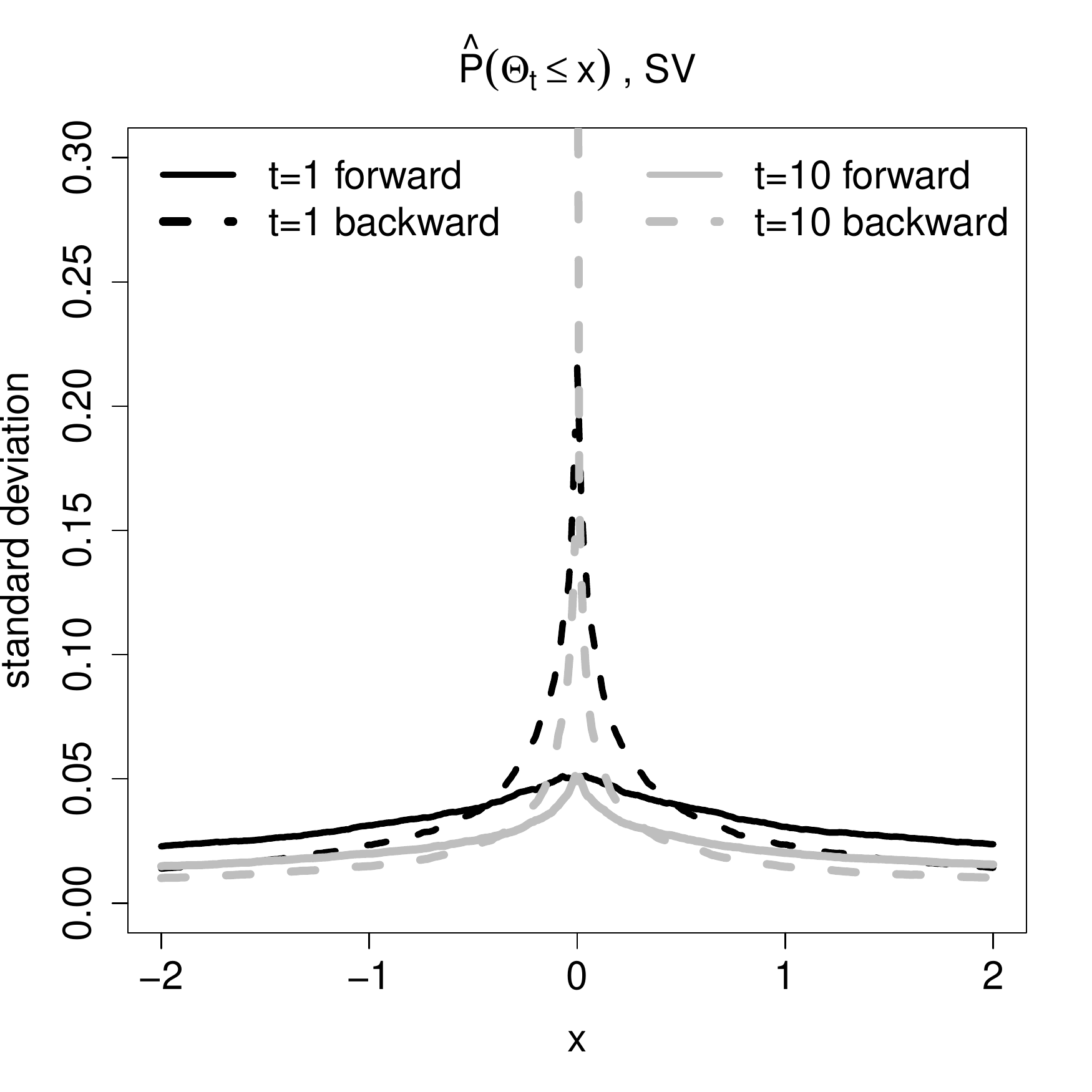}
&
\includegraphics[width=0.3\textwidth]{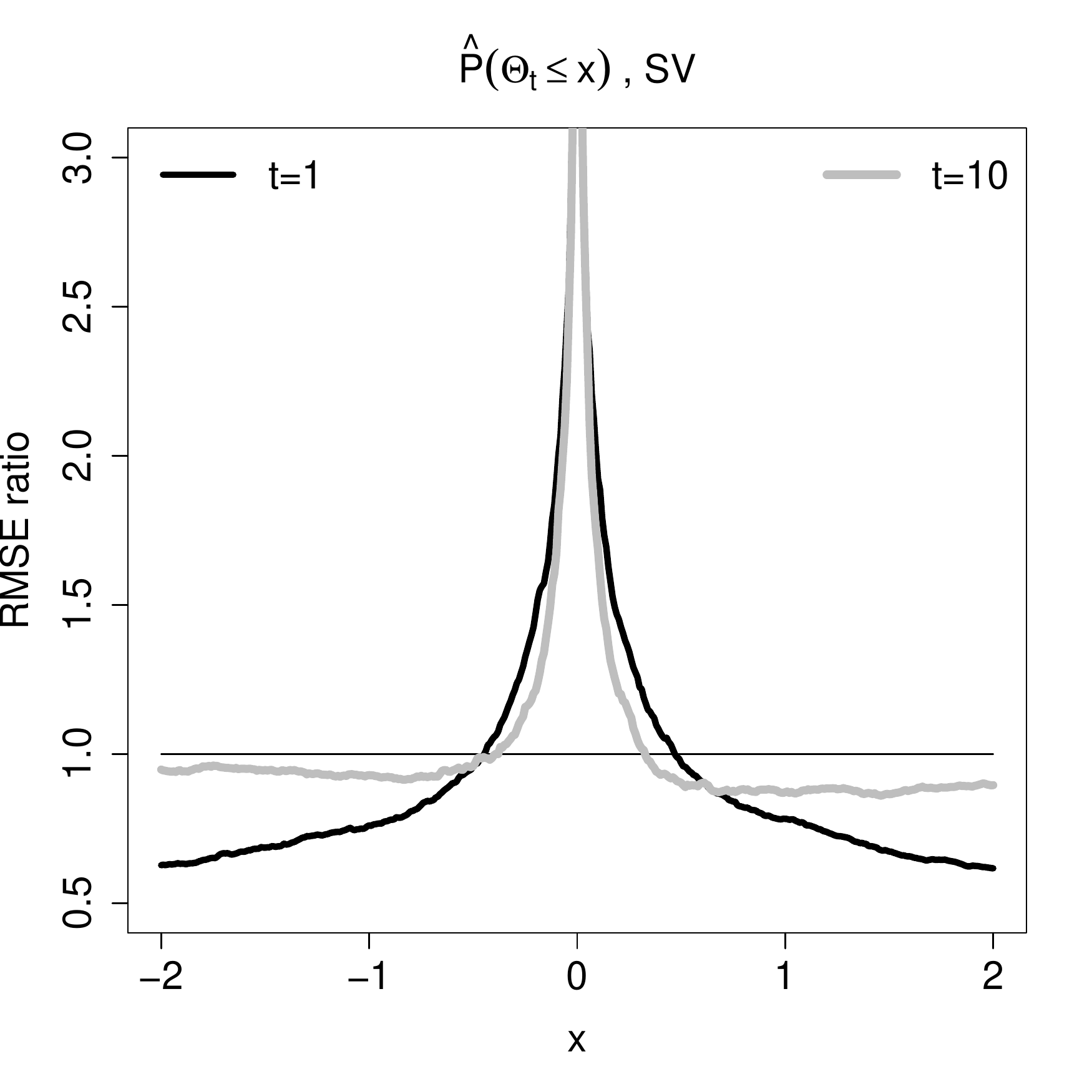}
\end{tabular}
\end{center}
\caption{\label{plot:for_back_rmse} \footnotesize
Performance of the forward and the backward estimators: bias (left), standard deviation (middle), ratio of root mean square errors (right) with respect to the pre-asymptotic values in the forward estimator, for GARCH(1,1) model (top) and SV model (bottom). }
\end{figure}

The right column, which shows the RMSE of the backward estimator divided by the RMSE of the forward estimators (both with respect to the pre-asymptotic values of the spectral tail process in the forward representation), shows that the backward estimator outperforms the forward estimator if $x$ is sufficiently large in absolute value. This phenomenon was also observed at other lags (not shown). For some other models, however, such as certain stochastic recurrence equation or copula Markov models, the advantage of the backward estimator was observed only for smaller lags ($t=1,\ldots,4$).

\subsection{Bootstrapped spectral tail process}
\label{sec:sim:bootstrap}

We asses the performance of the two bootstrap schemes, the stationary bootstrap and the multiplier block bootstrap. To do so, we estimate the coverage probability of the bootstrapped confidence intervals with respect to the true pre-asymptotic spectral tail process in the forward representation. We focus on probabilities of the form $\prob{\abs{\Theta_t}>1}$. This particular value can be of interest due to its interpretation as the probability of a shock being followed by an even larger aftershock, i.e., $\abs{X_t}$ being larger than $\abs{X_0}$ conditionally on $\abs{X_0}$ exceeding some threshold already. The true pre-asymptotic values in the forward representation were calculated numerically via $10\,000$ Monte Carlo simulations with time series of length $10\,000$.

\begin{figure}[htp]
\begin{center}
\begin{tabular}{cc}
\includegraphics[width=0.4\textwidth]{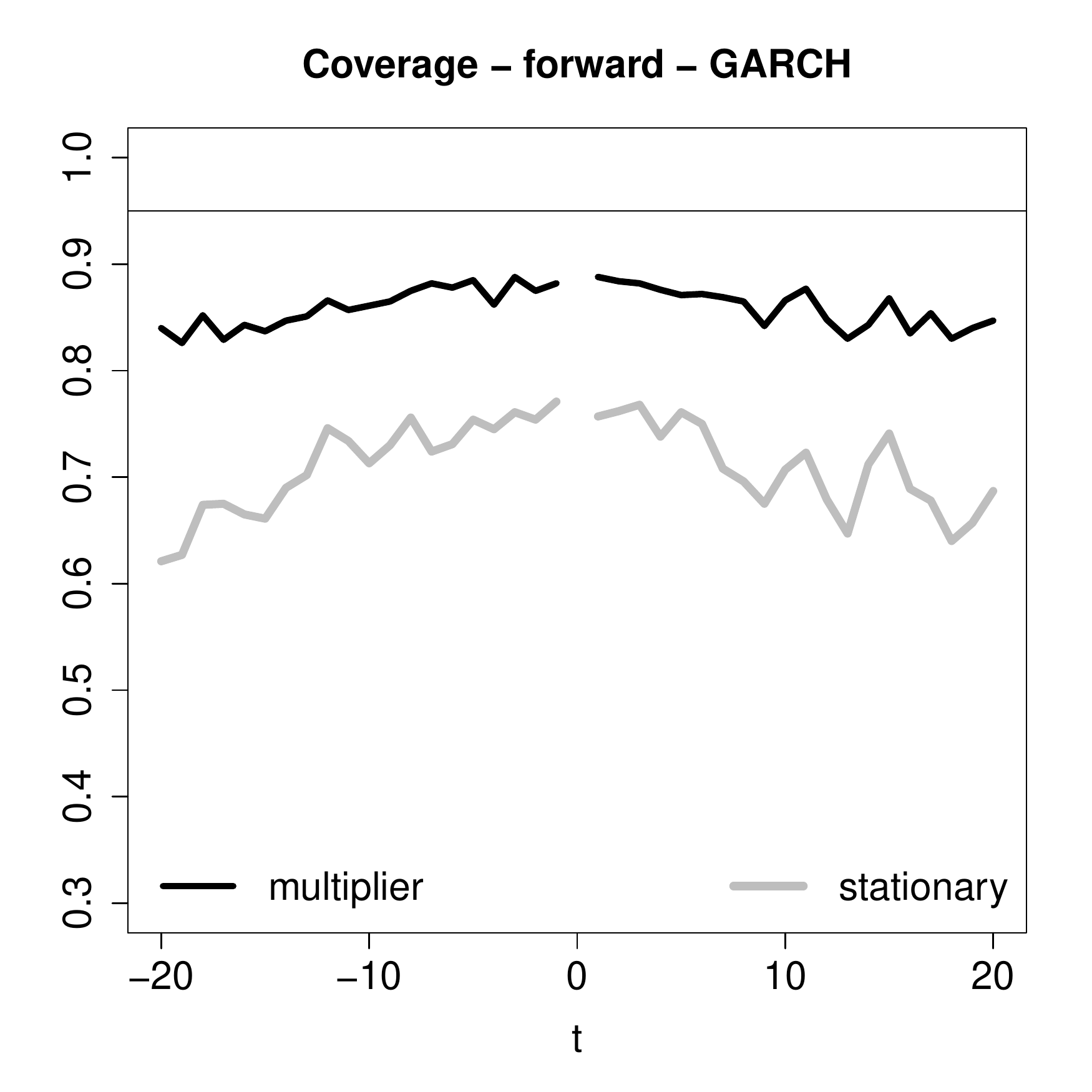}
&
\includegraphics[width=0.4\textwidth]{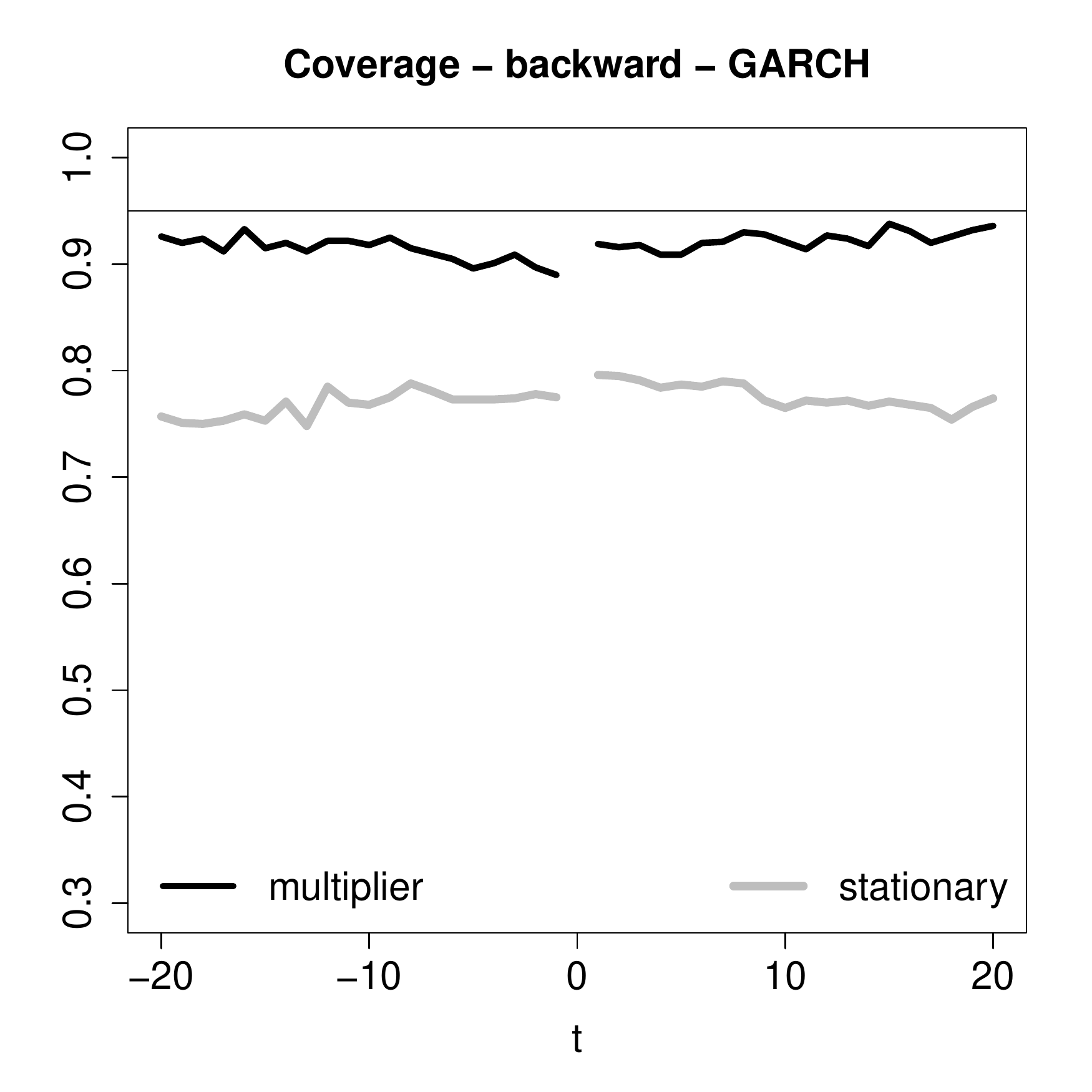}\\

\includegraphics[width=0.4\textwidth]{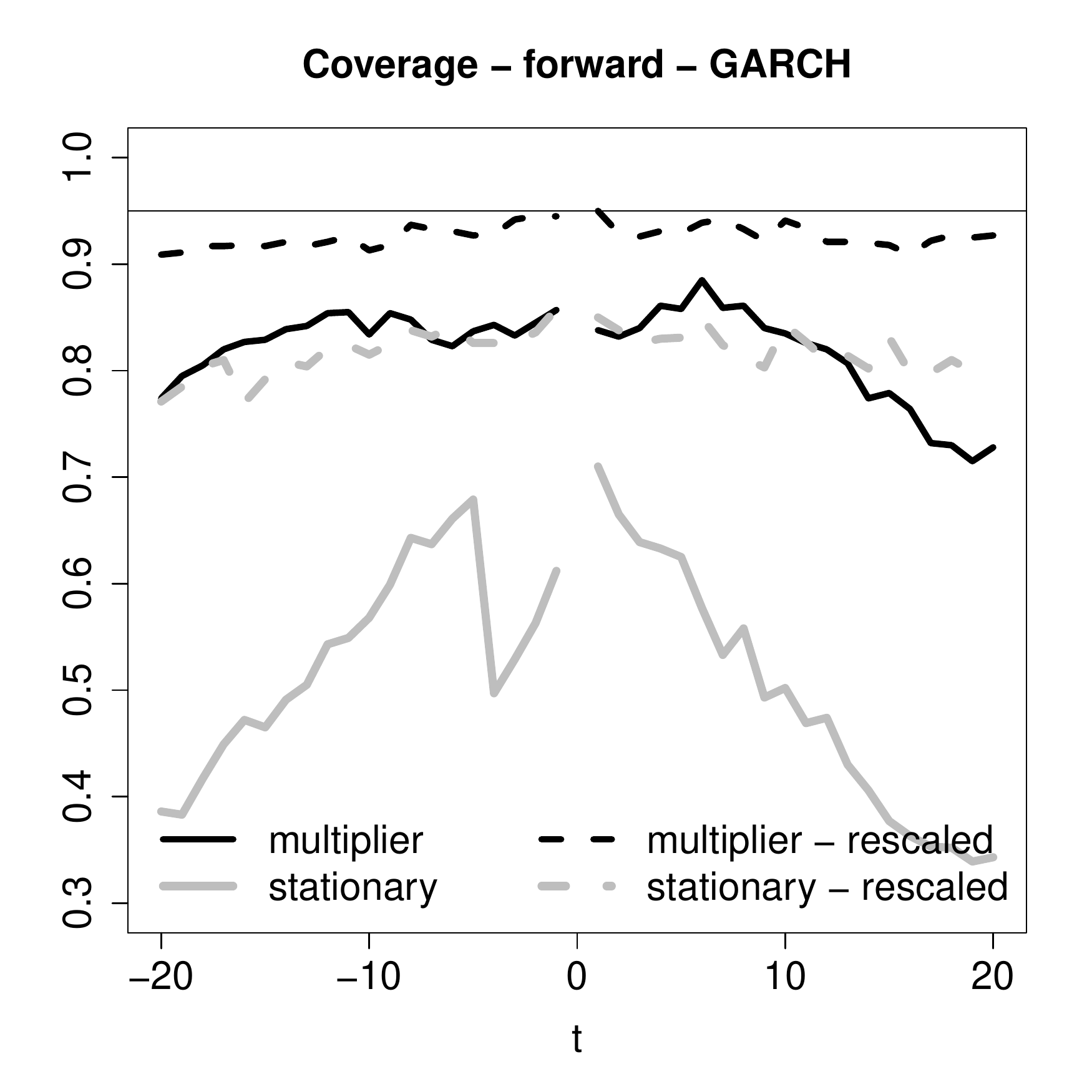}
&
\includegraphics[width=0.4\textwidth]{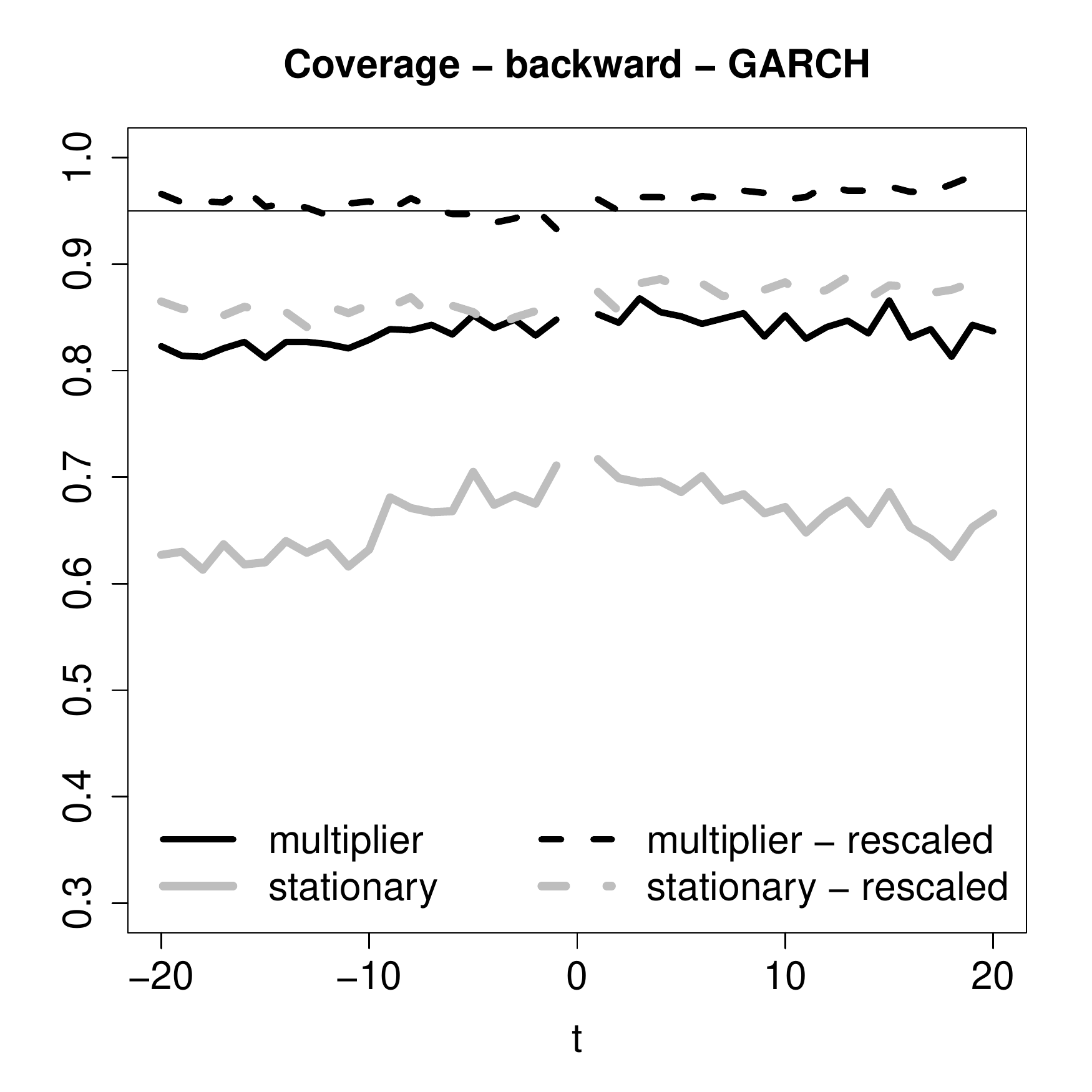}\\

\includegraphics[width=0.4\textwidth]{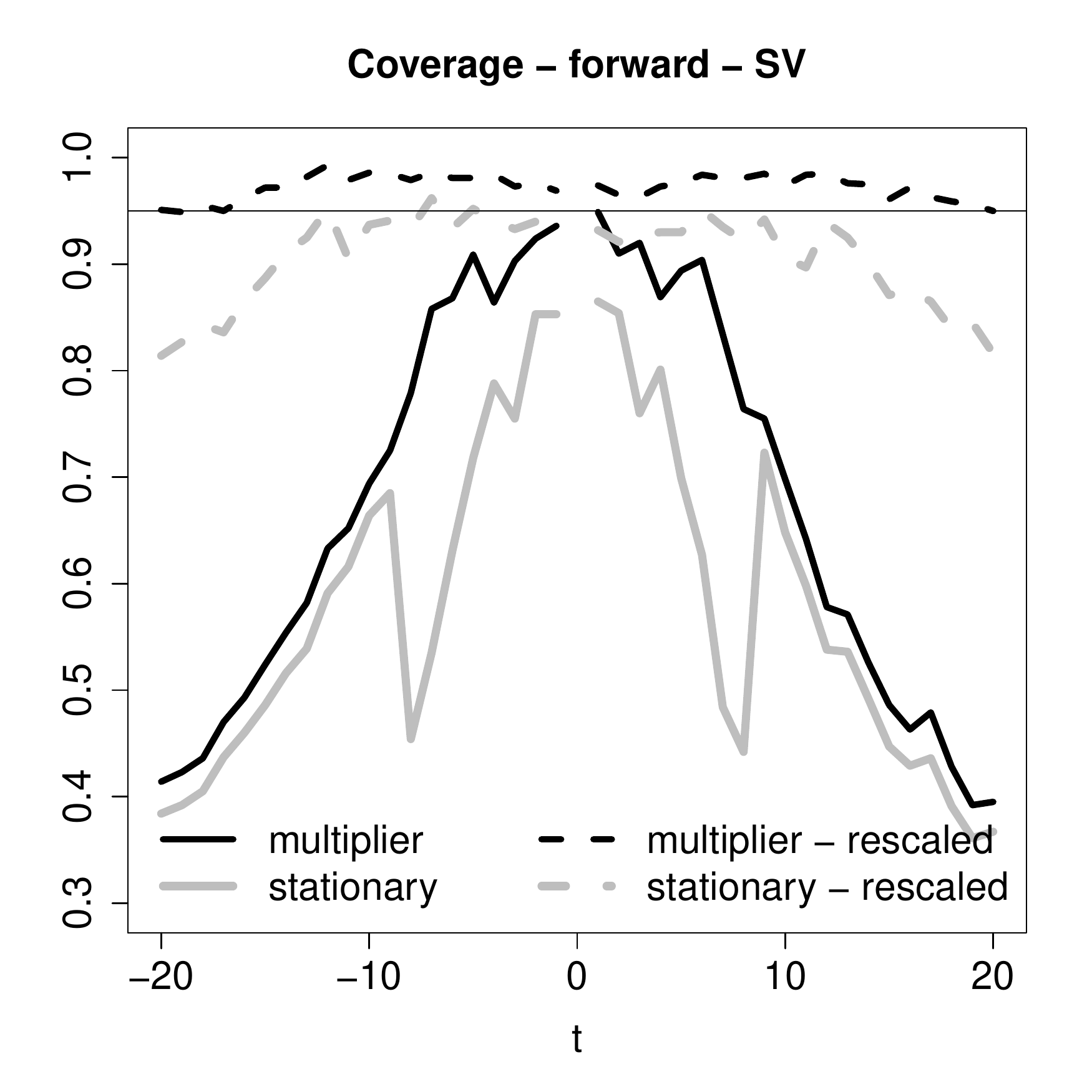}
&
\includegraphics[width=0.4\textwidth]{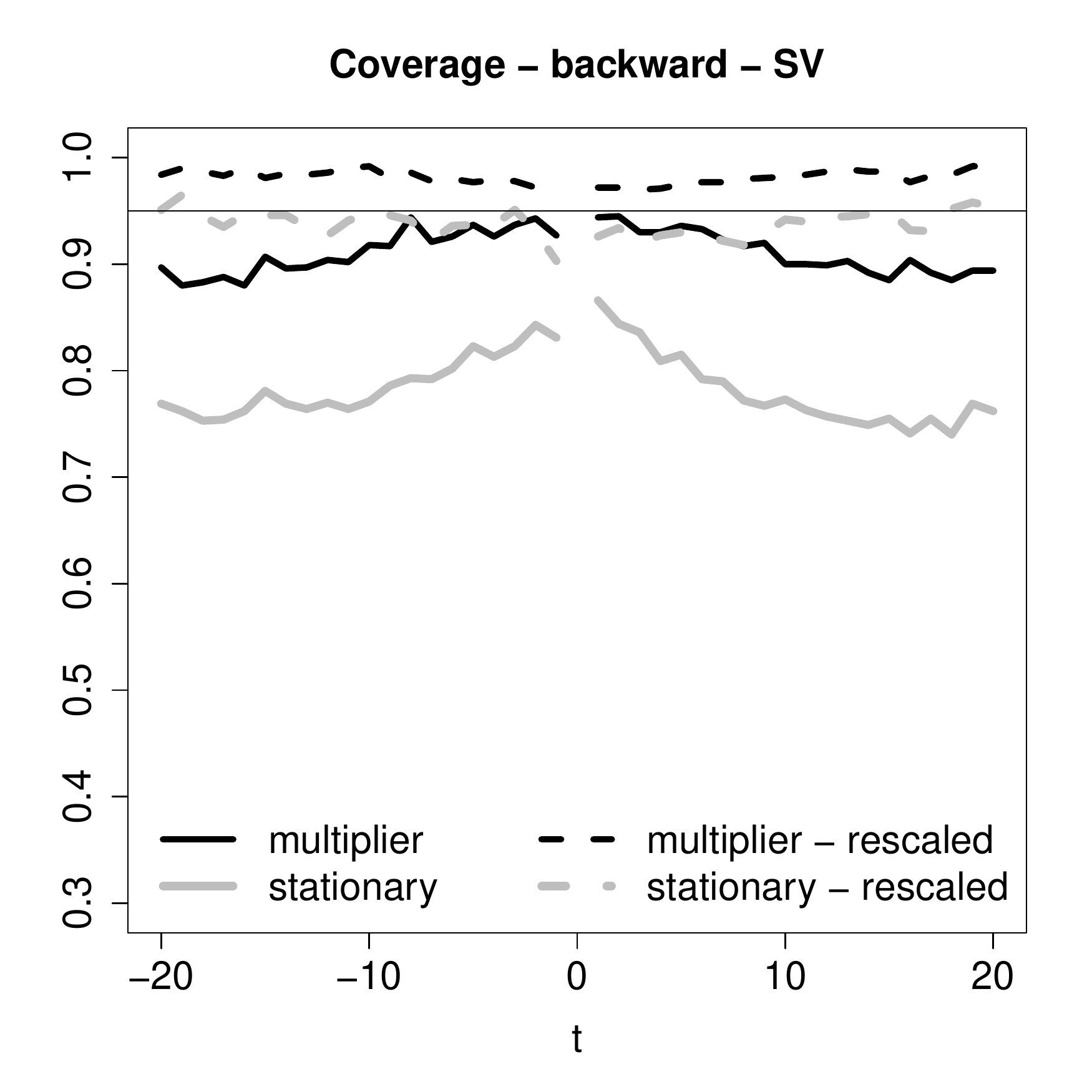}
\end{tabular}
\end{center}
\caption{\label{plot:boot} \footnotesize
Coverage probabilities of confidence intervals for $\prob{|X_t/X_0|>1\mid |X_0|>u_n}$ (left: forward estimator; right: backward estimator)  based on the stationary bootstrap (gray) and the multiplier block bootstrap (black). The top and the middle plots correspond to the GARCH$(1,1)$ model with thresholds set at the $95\%$ and the $98\%$ empirical quantiles, respectively. The bottom plots correspond to the SV simulation study with threshold set at the $98\%$ empirical quantile. In the latter two cases, the dashed lines correspond to the coverage probabilities of the rescaled confidence intervals \eqref{eq:rescalebootconf}. The horizontal black line is the $0.95$ reference line.}
\end{figure}

In Figure~\ref{plot:boot}, we plot the results for the GARCH$(1,1)$ model and for the SV model, for the forward and the backward estimators. The expected block size for the stationary bootstrap (represented by gray lines) was chosen as $100$. For the multiplier block bootstrap (black lines), the block size was fixed at $100$ and the multiplier variables $\xi_j$ were drawn independently from the standard normal distribution. Estimates of the coverage probabilities are based on $1\,000$ simulations. In each such sample, we use $1\,000$ bootstrap samples for calculating the confidence intervals with nominal coverage probability $95\%$. We use two different thresholds, i.e., the $95\%$ and $98\%$  empirical quantiles of the absolute values of a time series of length $n=2\,000$. For the higher threshold, the confidence intervals were calculated either directly (indicated by the solid lines) or using a rescaled bootstrap estimator that was based on the exceedances over the $95\%$ empirical quantile as in \eqref{eq:rescalebootconf} (dashed lines).

\begin{figure}[tp]
\begin{center}
\begin{tabular}{cc}
\includegraphics[width=0.4\textwidth]{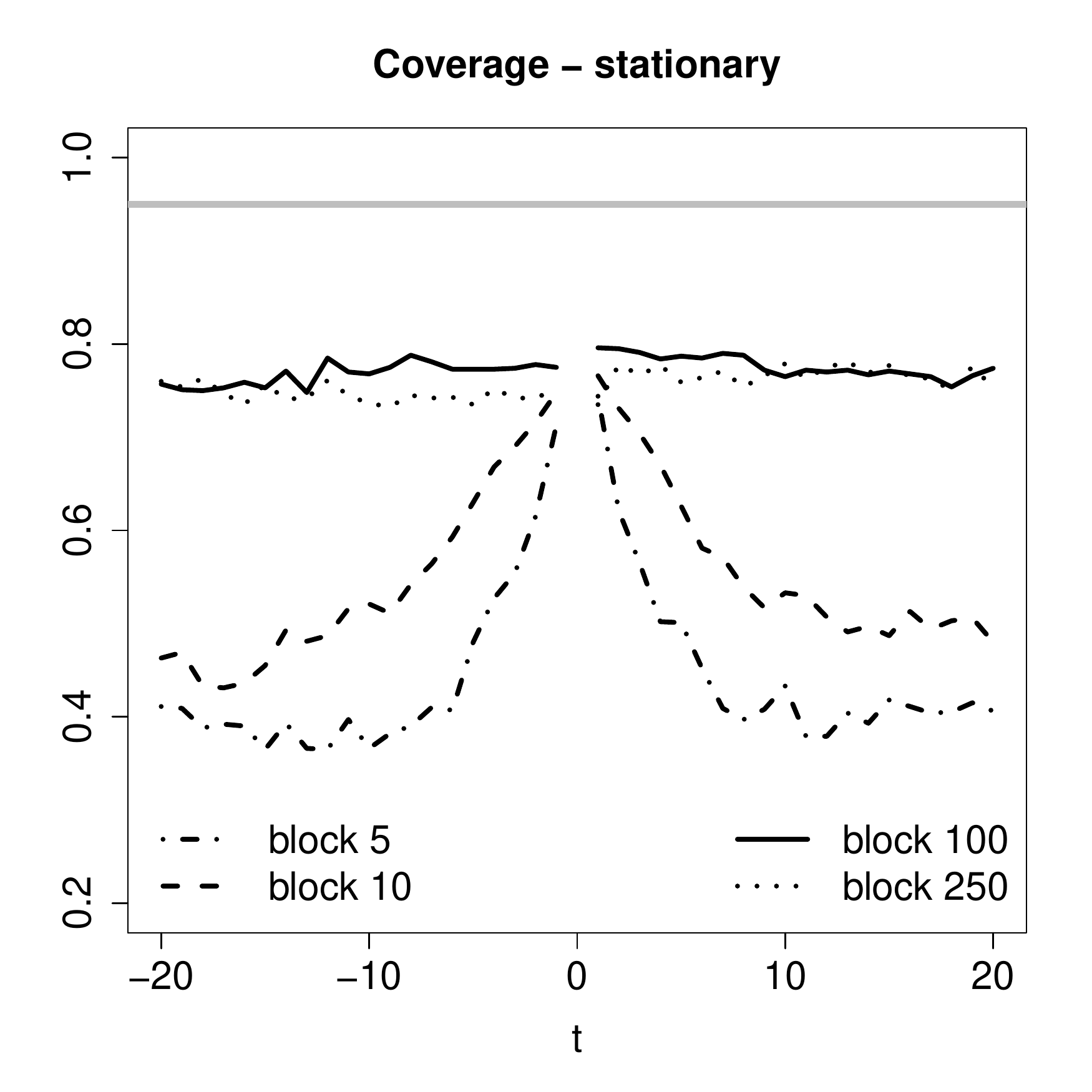}
&
\includegraphics[width=0.4\textwidth]{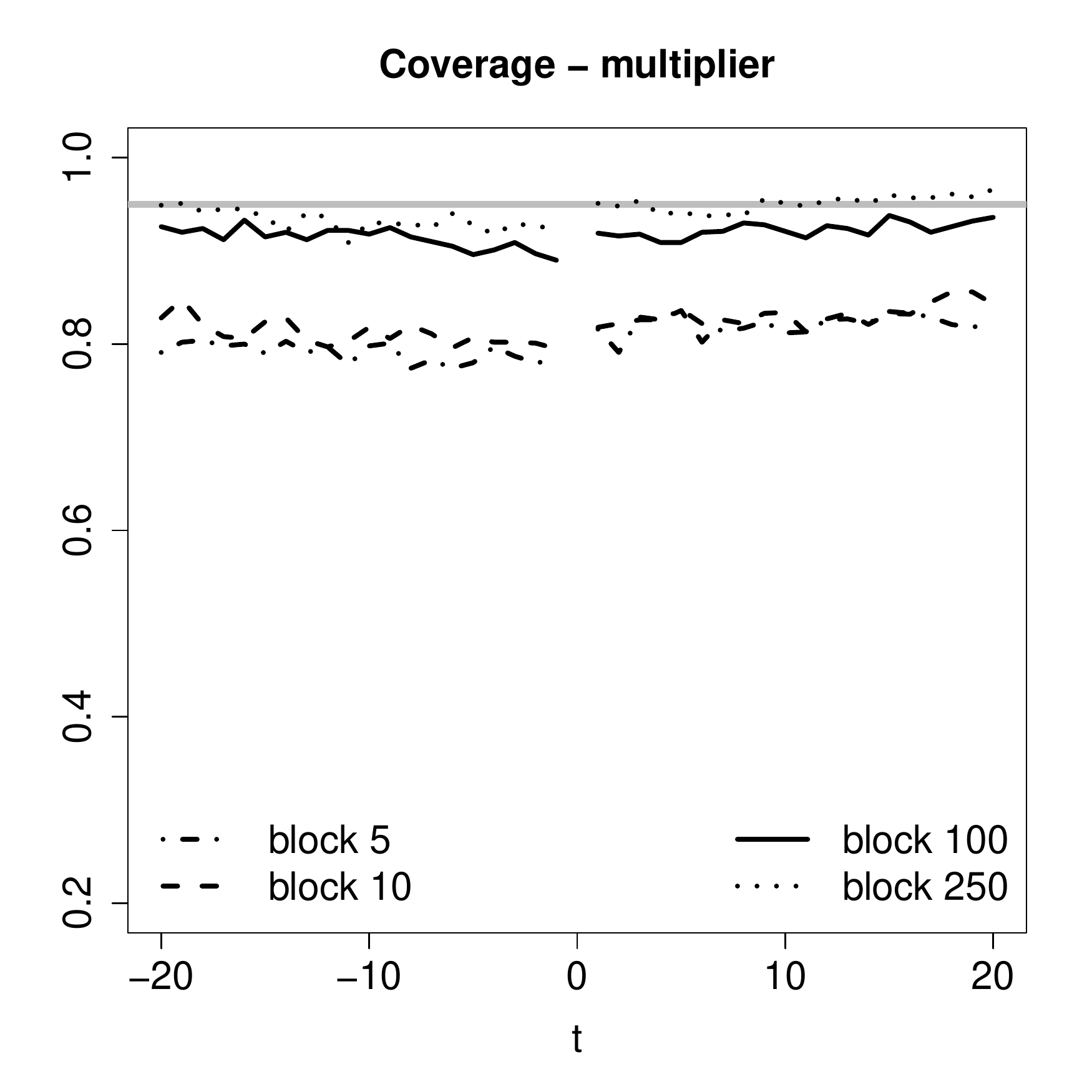}\\

\includegraphics[width=0.4\textwidth]{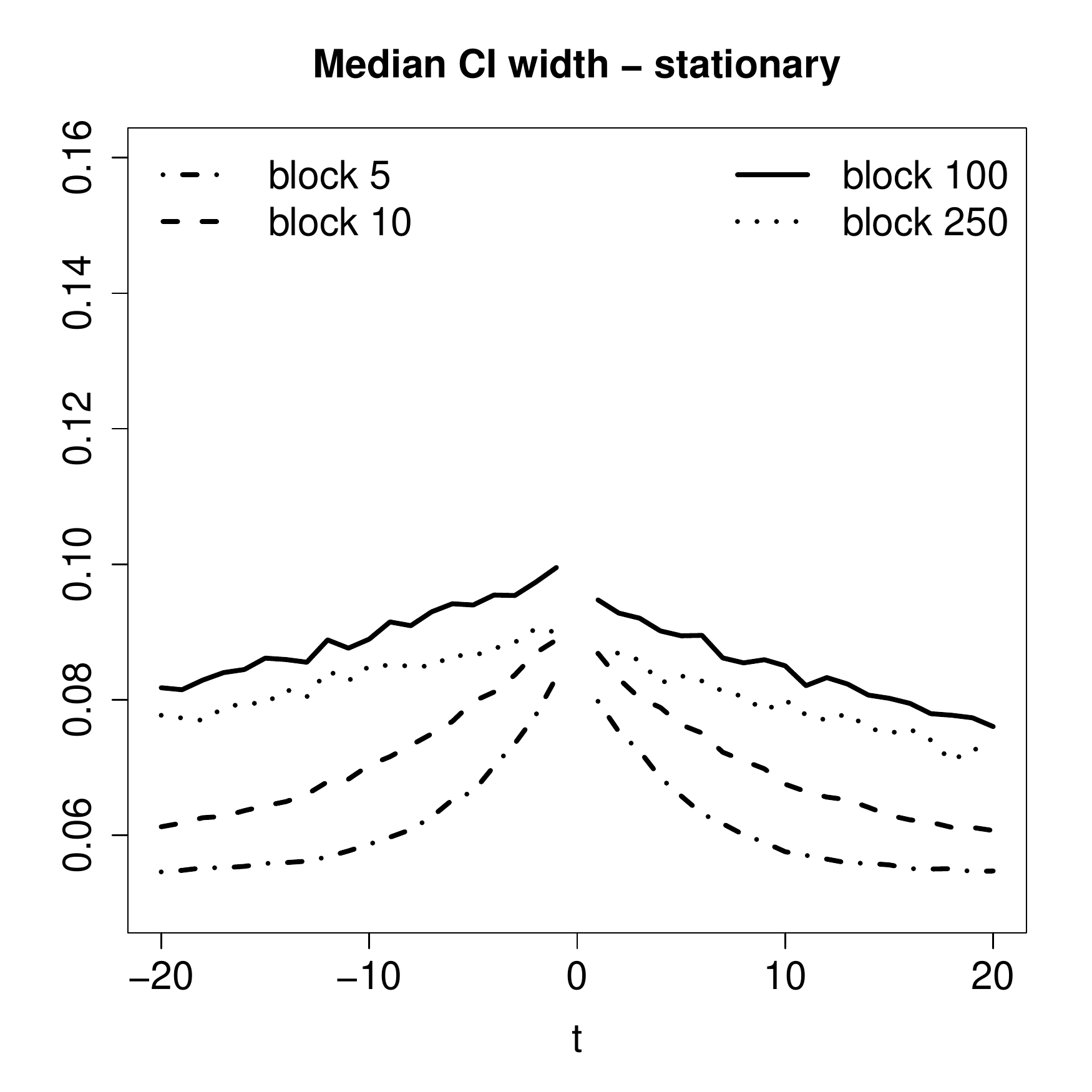}
&
\includegraphics[width=0.4\textwidth]{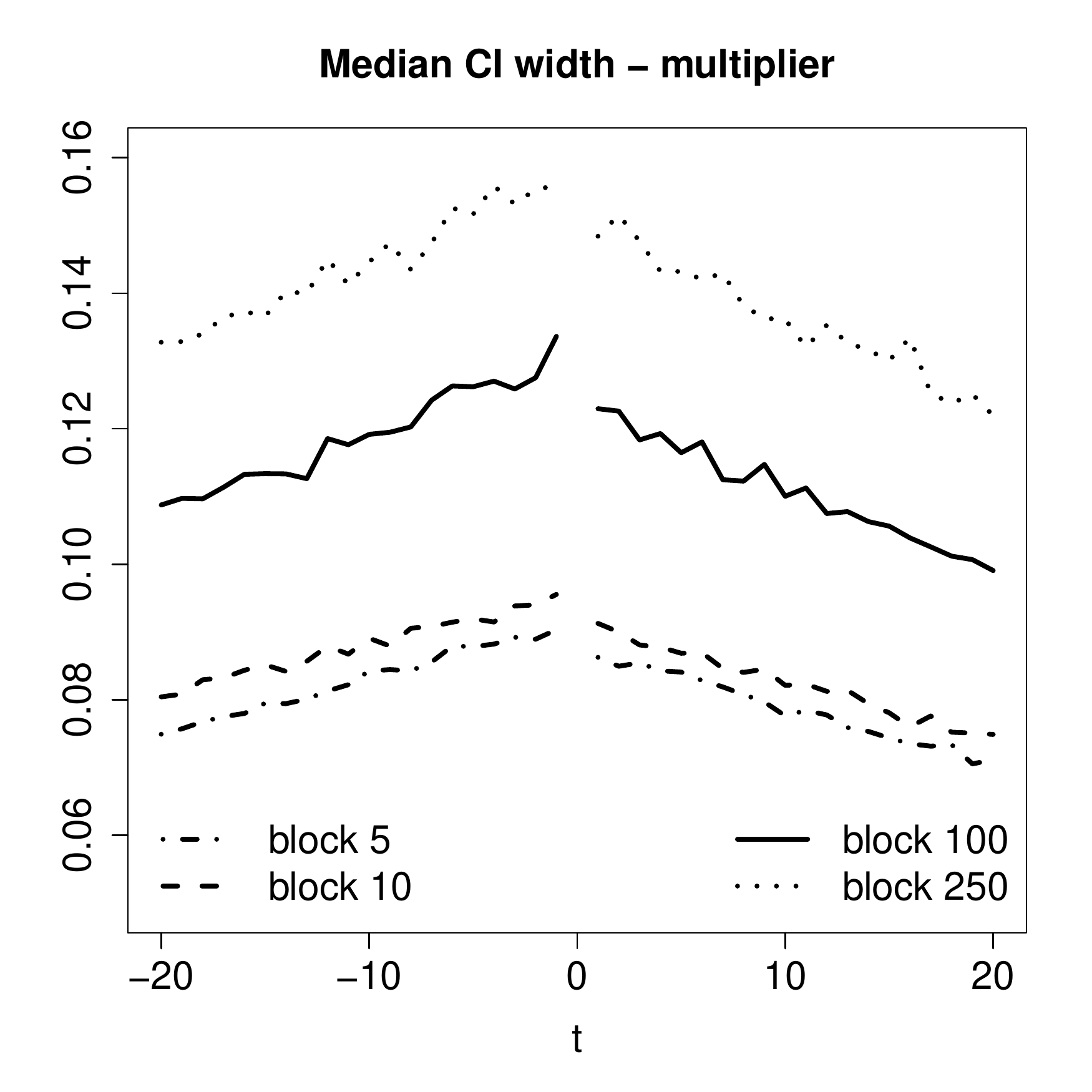}
\end{tabular}
\end{center}
\caption{\label{plot:boot:block:cov} \footnotesize
Coverage probabilities (top) and median widths (bottom) of confidence intervals for $\prob{ \lvert X_t/X_0 \rvert >1\mid |X_0|>u_n}$ based on the stationary bootstrap (left column) and the multiplier block bootstrap (right column) for different block lengths. The dash--dotted, dashed, solid, and dotted lines represent (mean) block lengths $5$, $10$, $100$, and $250$ respectively. The plots correspond to the backward estimator and the GARCH$(1,1)$ model with thresholds set at the $95\%$ empirical quantile. The horizontal black lines in the top plots are the $0.95$ reference lines.}
\end{figure}

In all cases, the multiplier block bootstrap produces a better coverage probability than the stationary bootstrap. Moreover, the backward estimator is more stable than the forward one, at least for $x=1$, and this translates into higher stability of the  bootstrapped confidence intervals. The effect is especially visible for higher thresholds, e.g., at the $98\%$ quantile, leaving insufficiently many pairs of exceedances for accurate inference. Finally, rescaled confidence intervals \eqref{eq:rescalebootconf} based on lower thresholds can have a much better coverage than confidence intervals based on higher thresholds.

 In addition, in Figure~\ref{plot:boot:block:cov} we show coverage probabilities and median confidence intervals widths for different block sizes. The multiplier block bootstrap is more robust to the choice of block length than the stationary bootstrap. In contrast to the stationary bootstrap, the multiplier block bootstrap produces confidence intervals whose coverage probabilities are fairly stable across different lags for a given block length.

 It is important not to set the block length too low since it can lead to poor coverage probabilities, especially for higher lags. On the other hand, too large a block length can result in confidence intervals that are too wide.

\section{Application}
\label{sec:application}
We first consider daily log-returns on the S\&P500 stock market index  between 1990-01-01 and 2010-01-01 taken from Yahoo Finance\footnote{\url{http://finance.yahoo.com/}}. In Figure~\ref{plot:app1}, we plot the sample spectral tail process probabilities $\prob{ \abs{\Theta_t} > 1 \mid \Theta_0 = \pm 1 }$ and $\prob{ \pm \Theta_t > 1 \mid \Theta_0 = \pm 1 }$ based on the backward estimator with $98\%$ empirical quantile taken as a threshold and the $80\%$ pointwise confidence intervals from the multiplier bootstrap scheme rescaled via the $95\%$ quantile as threshold as in \eqref{eq:rescalebootconf}. The estimated index of regular variation is $\hat{\alpha}=3.17$.

The left-hand plots correspond to conditioning on a positive extreme at the current time instant, whereas the right-hand plots correspond to conditioning on a negative shock. The former plots indicate much weaker serial extremal  dependence than the latter ones: negative shocks are more persistent than positive ones. This is indicated by the lower bounds of the confidence intervals being above the horizontal lines which correspond to probabilities under independence. In particular, the pattern of negative extremes followed by positive ones is clearly visible; see the right-hand plot on the second row. The above mentioned characteristics are shared by other stock's daily returns which were tested but not reported here.

\begin{figure}[htp]
\begin{center}
\begin{tabular}{cc}
\includegraphics[width=0.4\textwidth]{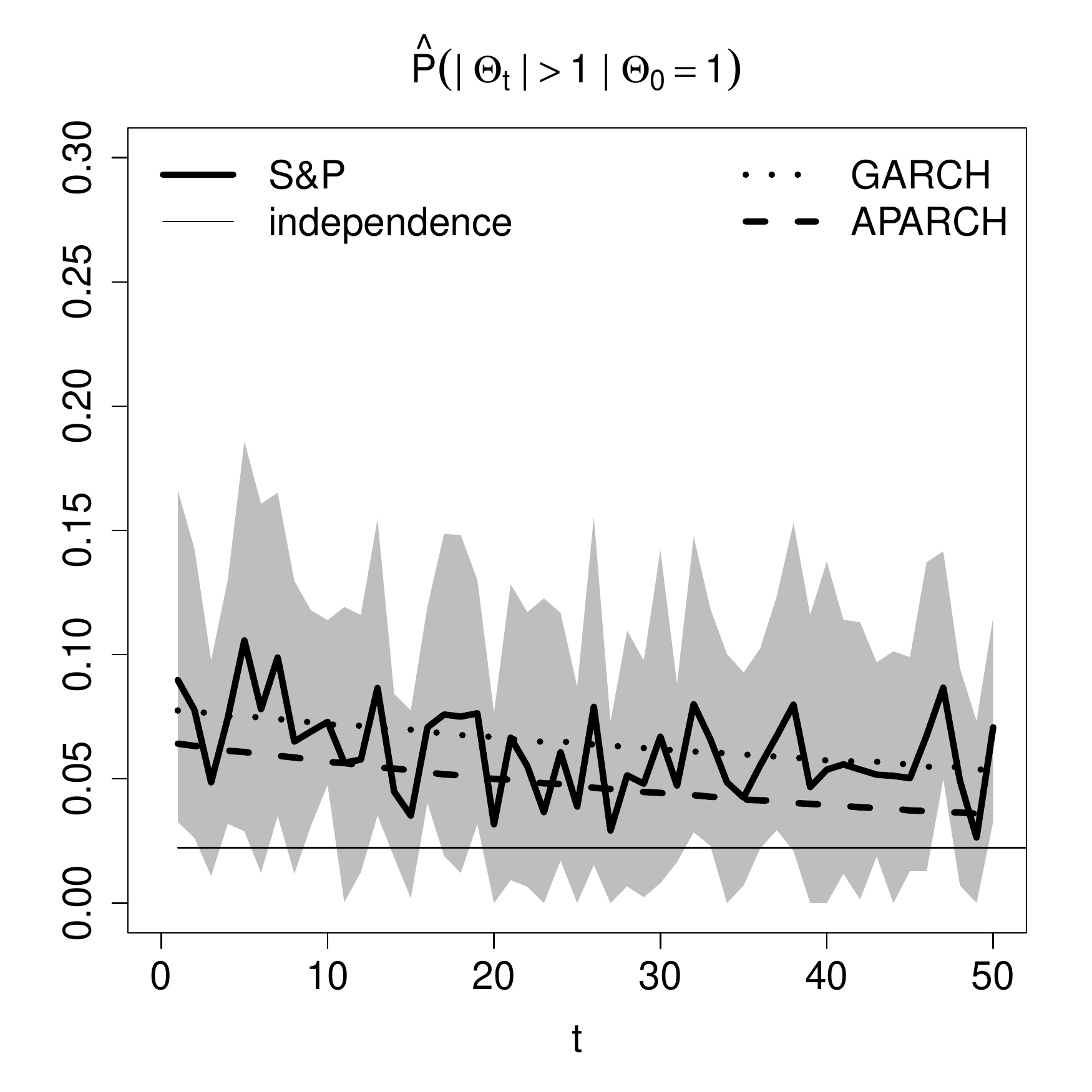}
&
\includegraphics[width=0.4\textwidth]{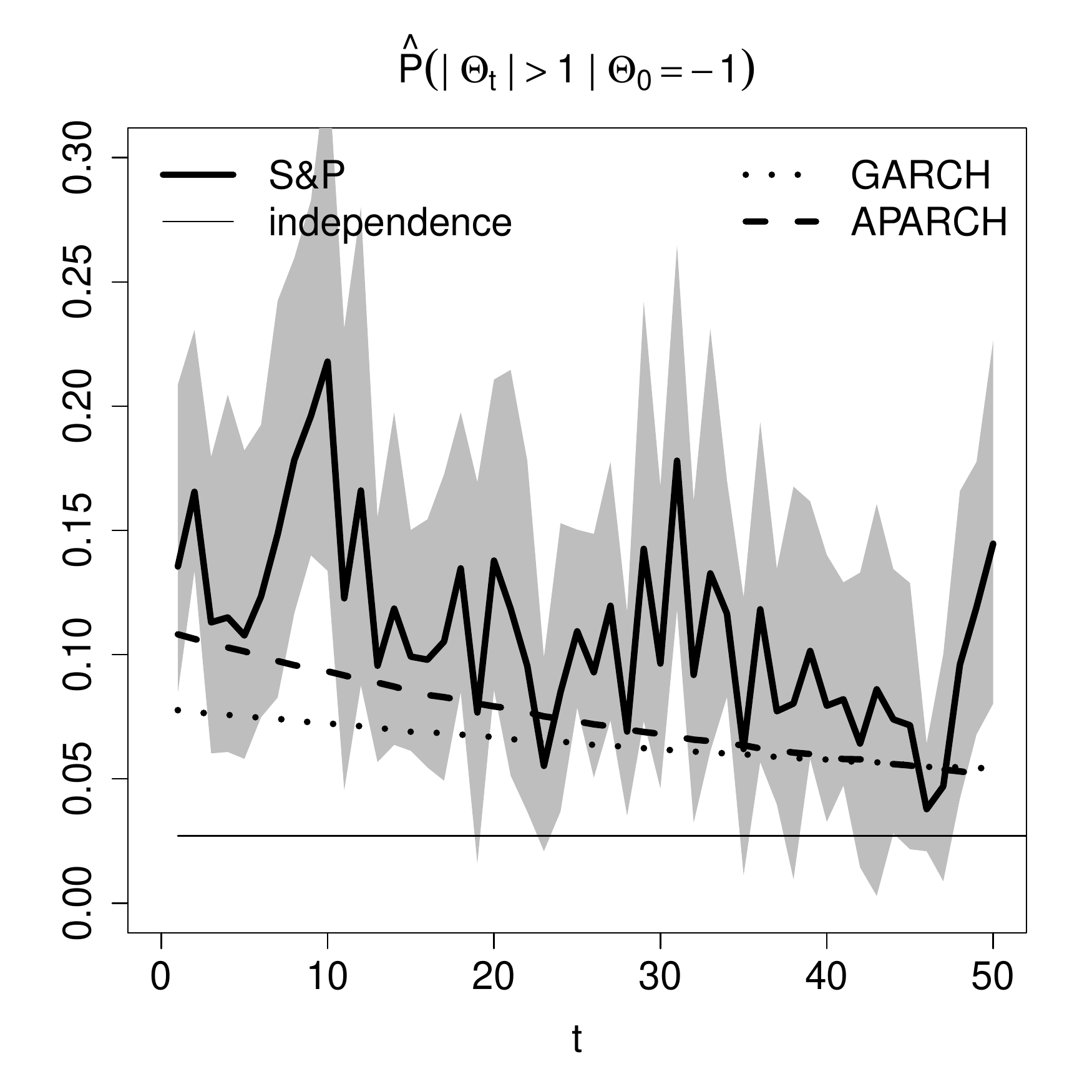}\\
\includegraphics[width=0.4\textwidth]{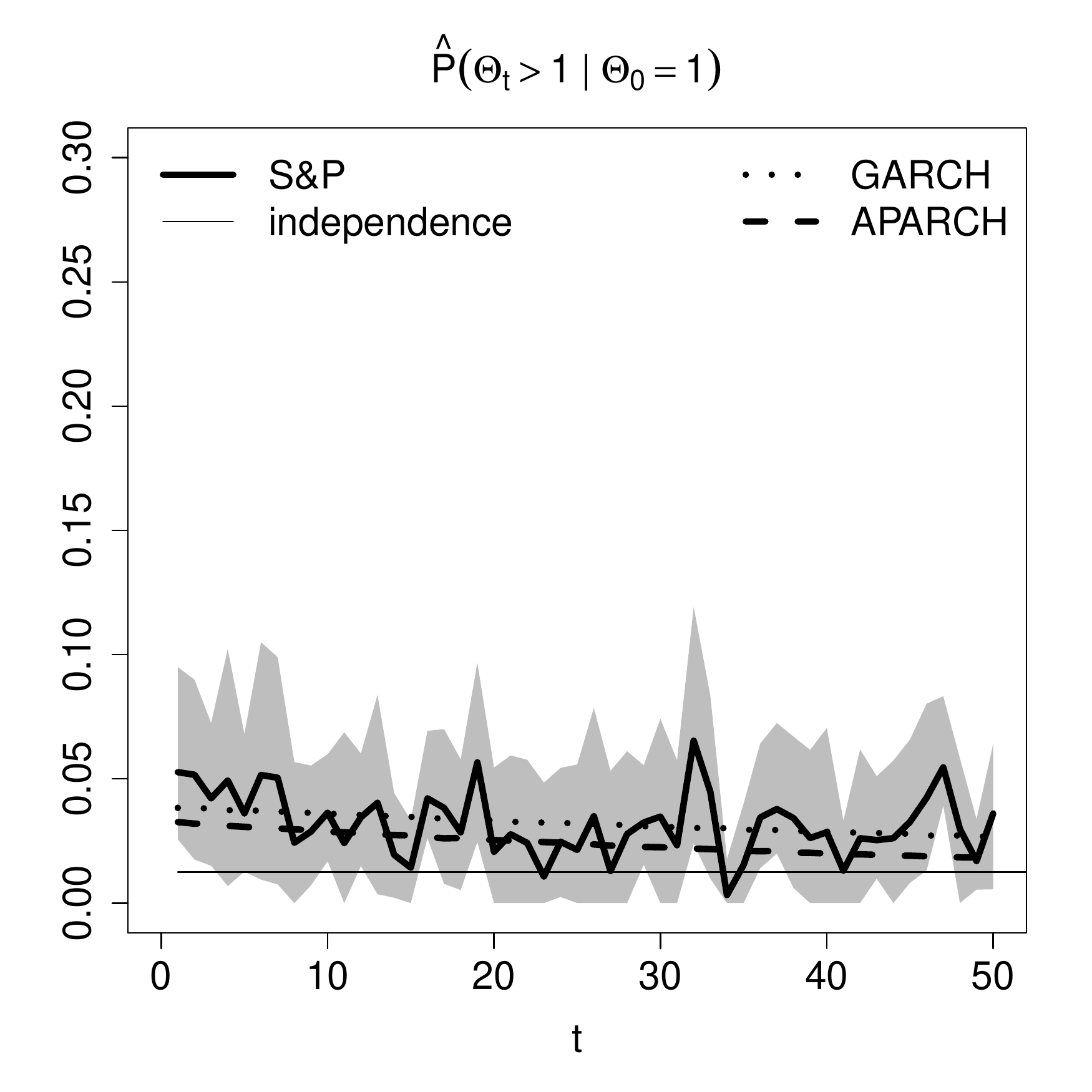}
&
\includegraphics[width=0.4\textwidth]{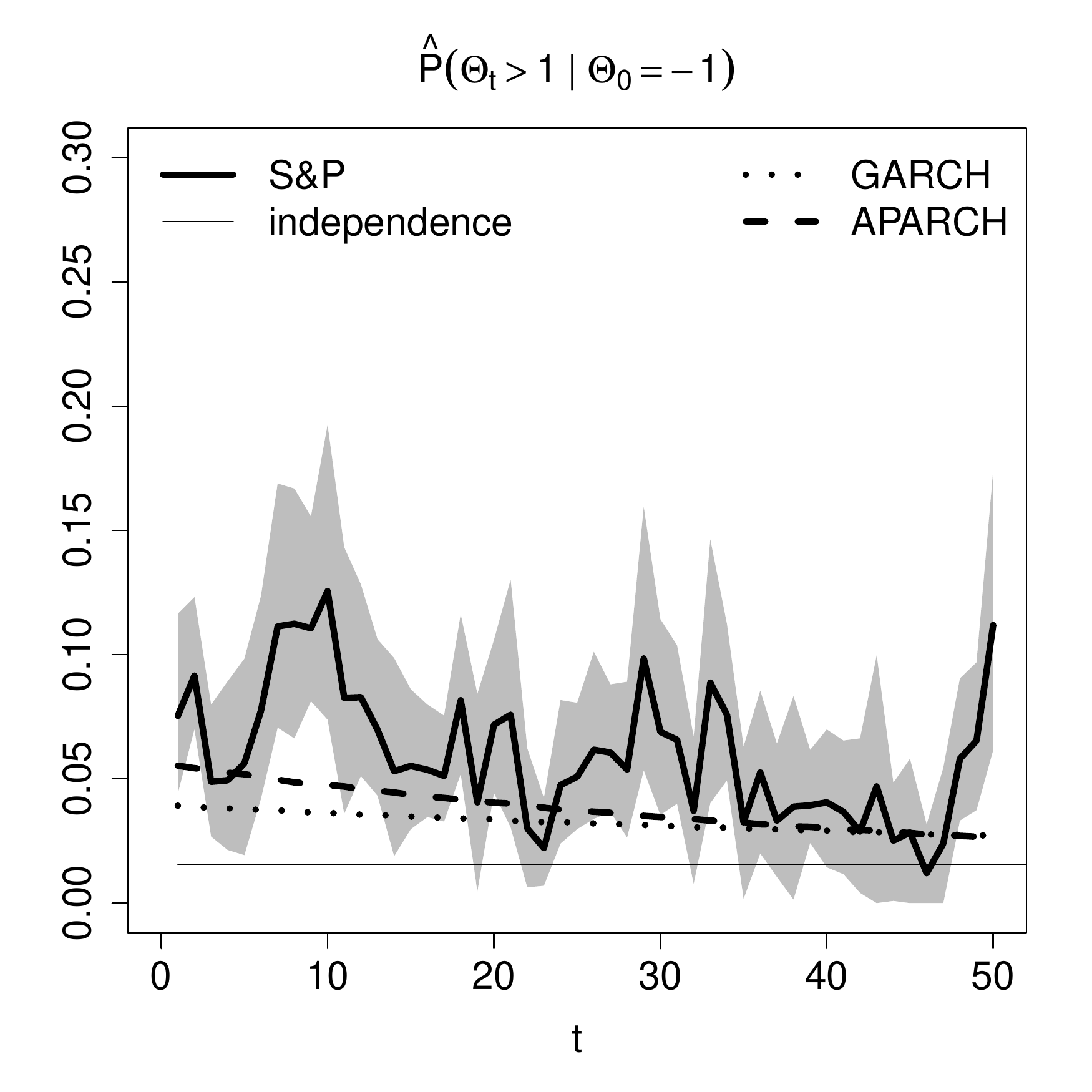}\\
\includegraphics[width=0.4\textwidth]{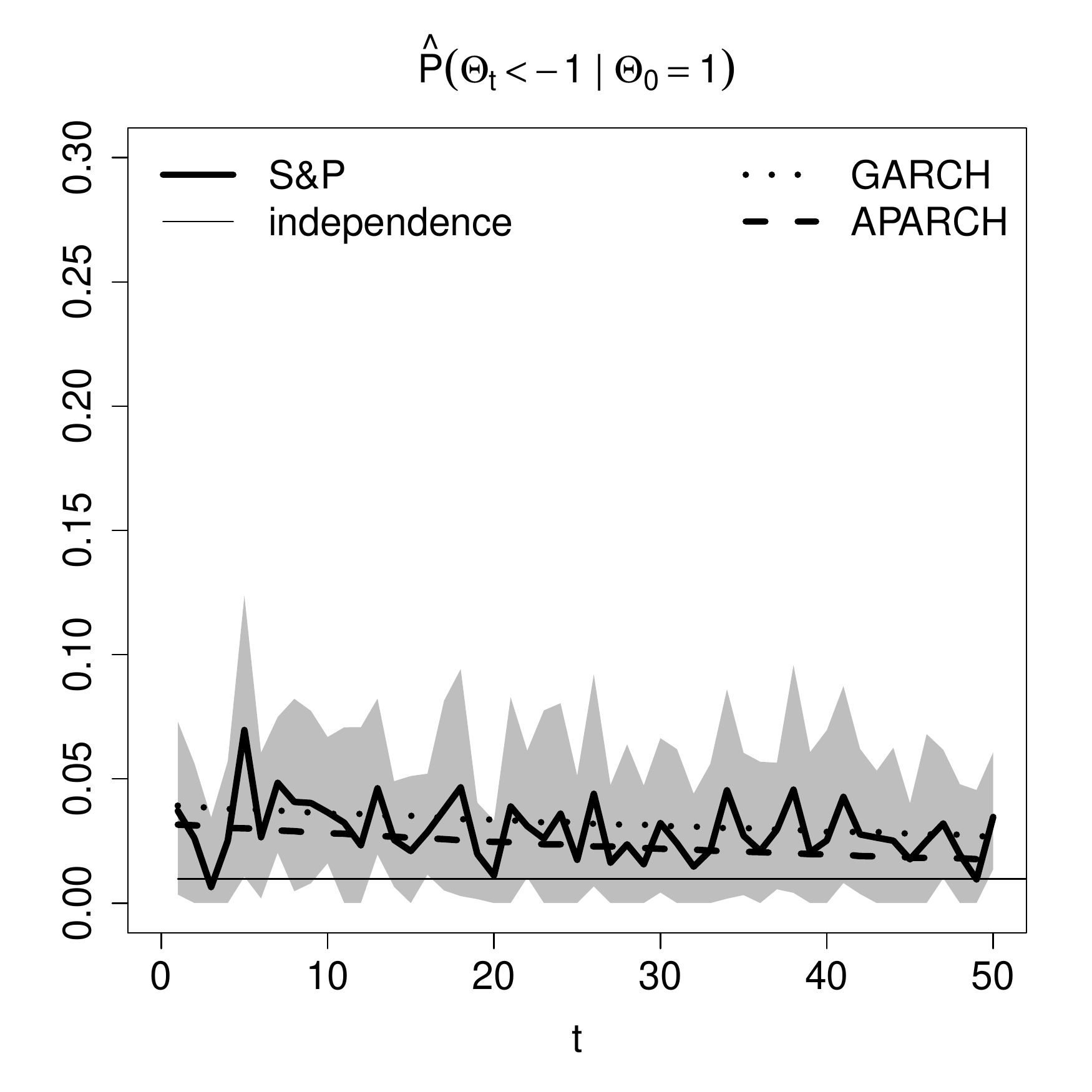}
&
\includegraphics[width=0.4\textwidth]{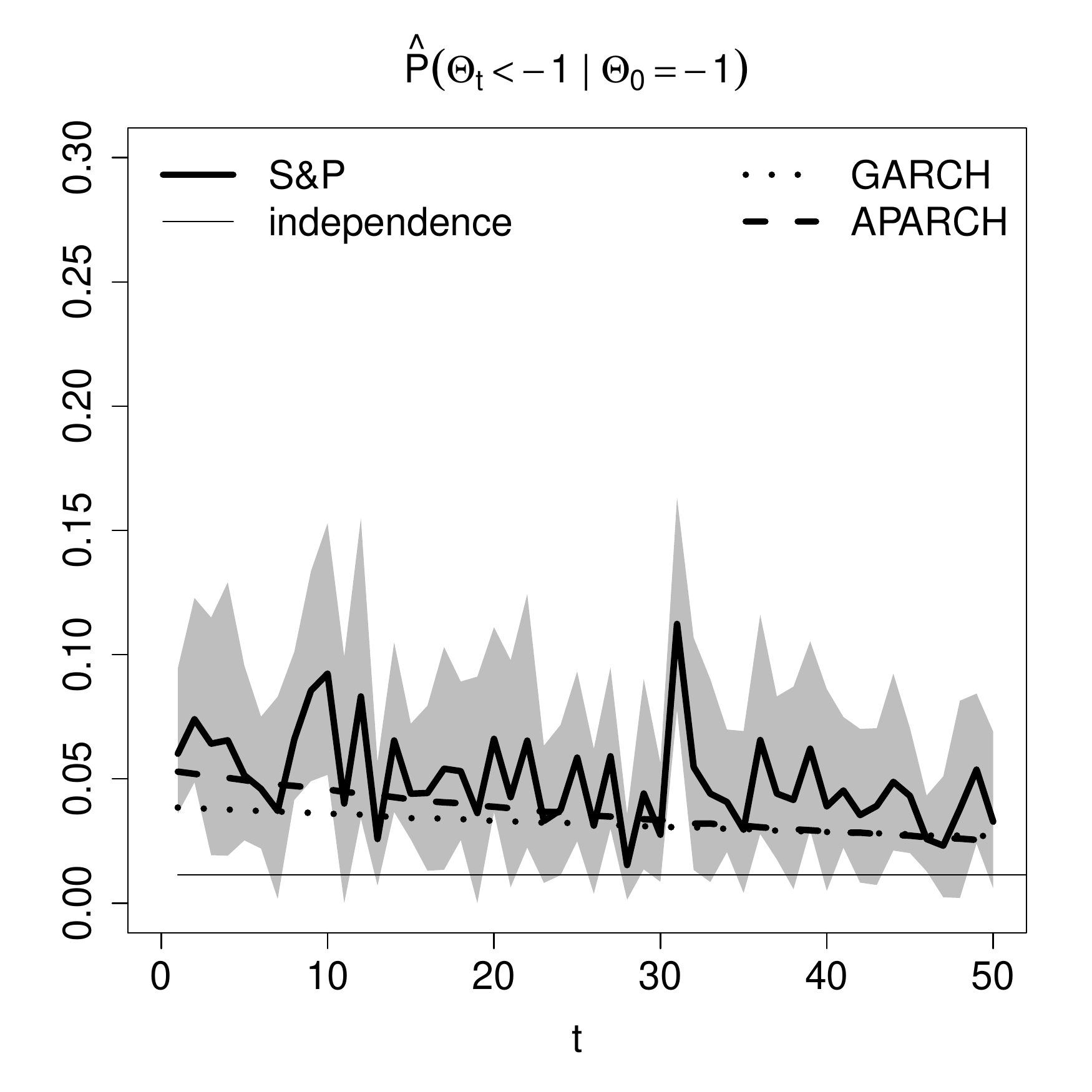}
\end{tabular}
\end{center}
\caption{\label{plot:app1} \footnotesize Sample spectral tail process probabilities (solid black bold line) for the S\&P500 daily log-returns based on the backward estimator  and the pre-asymptotic spectral tail process probabilities of the fitted GARCH(1,1) (dotted line) and APARCH(1,1) (dashed line) models. The gray area corresponds to the $80\%$ pointwise confidence intervals for the pre-asymptotic spectral tail probabilities based on the multiplier bootstrap with $1\,000$ replications. 
The top, middle and bottom rows concern the conditional probabilities that $|\Theta_t| > 1$, $\Theta_t > 1$ and $\Theta_t < -1$, respectively, given that $\Theta_0=1$ (left column) and $\Theta_0=-1$ (right column).
The horizontal line corresponds to the pre-asymptotic spectral tail probabilities under independence.}
\end{figure}

Consider two widely used financial models of the type $X_t=\sigma_t Z_t$: first, the GARCH$(1,1)$ process, where
\[
  \sigma_t^2=\omega + \alpha_1 X_{t-1}^2 + \beta_1 \sigma_{t-1}^2,
\]
and second, the APARCH(1,1) process \citep{Ding199383} with
\[
  \sigma_t^\delta
  =
  \omega + \alpha_1 \left(\abs{X_{t-1}} -\gamma _1 X_{t-1}\right)^{\delta} + \beta_1 \sigma_{t-1}^{\delta}.
\]
Both models allow for volatility clustering in the limit. Additionally, the APARCH model  captures asymmetry in the volatility of returns. That is, volatility tends to increase more when returns are negative, as compared to positive returns of the same magnitude if $\gamma_1>0$. The asymmetric response of volatility to positive and negative shocks is well known in the finance literature as the \emph{leverage effect} of the stock market returns \citep{Black1976}.

We fit those two models to daily log-returns of the S\&P500 index. We use the \textsf{garchFit} function from the \textsf{fGarch} library available in \textsf{R}, the function being based on maximum likelihood estimation \citep{fGarch}.
The innovations, $Z_t$, are assumed to be standard normally distributed. The fitted parameters are given in the top part of Table~\ref{tab:fit}.

\begin{table}
\begin{center}
\begin{tabular}{ l l  c  c  c  c  c  }
\toprule
  & & $\omega$ & $\alpha_1$  & $\beta_1$ & $\delta$ & $\gamma_1$  \\ \midrule
S\&P500 & GARCH & $7 \times 10^{-7}$  & $0.062$ & $0.932$ & - & -  \\
 & & $(2 \times 10^{-7})$ & $(0.006)$  & $(0.007)$  &  &  \\[1ex] 
 & APARCH & $5 \times 10^{-5}$ & $0.056$ & $0.937$ & $1.227$ & $0.874$  \\
  & & $(1 \times 10^{-5})$ & $(0.008)$  & $(0.006)$  & $(0.131)$ &  $(0.118)$\\
  \midrule
P\&G & GARCH & $9 \times 10^{-7}$  & $0.04$ & $0.957$ & - & -  \\
 & & $(2 \times 10^{-7})$ & $(0.004)$  & $(0.004)$  &  &  \\[1ex] 
 & APARCH & $17 \times 10^{-5}$ & $0.056$ & $0.951$ & $0.938$ & $0.608$  \\
   & & $(3 \times 10^{-5})$ & $(0.004)$  & $(0.004)$  & $(0.112)$ &  $(0.074)$\\
\bottomrule
\end{tabular}
\end{center}
\caption{\label{tab:fit}\footnotesize Parameters of the models fitted to daily log-returns of the S\&P500 index (top) and the P\&G stock price (bottom). Standard errors in parentheses.}
\end{table}

In Figure~\ref{plot:app1} we plot the pre-asymptotic spectral tail process probabilities based on the forward estimator for the fitted GARCH and APARCH models, together with the sample spectral tail process probabilities for S\&P500 daily log-returns estimated by the backward estimator. The pre-asymptotic values corresponding to the fitted models were calculated numerically via $10\,000$ Monte Carlo simulations with time series of length $10\,000$. Clearly, the APARCH model captures the asymmetry which the GARCH model cannot.

As a second example, we study daily log-returns on the P\&G stock price between 1990-01-01 and 2010-01-01. The tail index is estimated at $\hat{\alpha}=3.3$. We fit the GARCH(1,1) and APARCH(1,1) models to the time series and show the estimated parameters in the bottom part of Table~\ref{tab:fit}.

In Figure~\ref{plot:app3} we plot the sample spectral tail process probabilities based on the daily log-returns themselves and on the residuals of the fitted GARCH(1,1) and APARCH(1,1) models  obtained by the backward estimator. The top-right plot indicates that there is significant serial extremal dependence in the P\&G daily log-returns triggered by the negative shocks. Due to high asymmetry in volatility, this feature is still present in the residuals of the fitted GARCH model whereas it is better removed by the APARCH filter.

\begin{figure}[htp]
\begin{center}
\begin{tabular}{cc}
\includegraphics[width=0.4\textwidth]{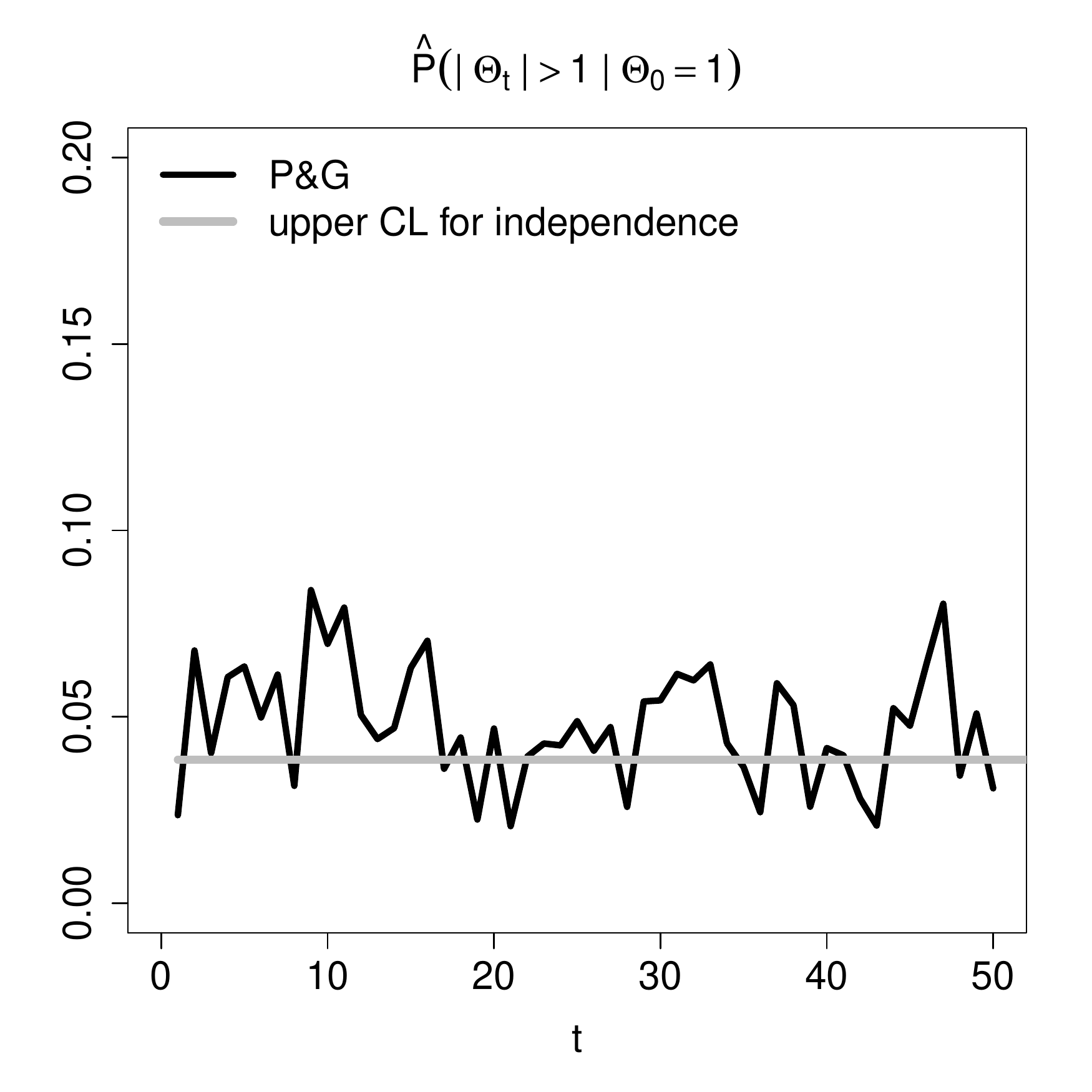}
&
\includegraphics[width=0.4\textwidth]{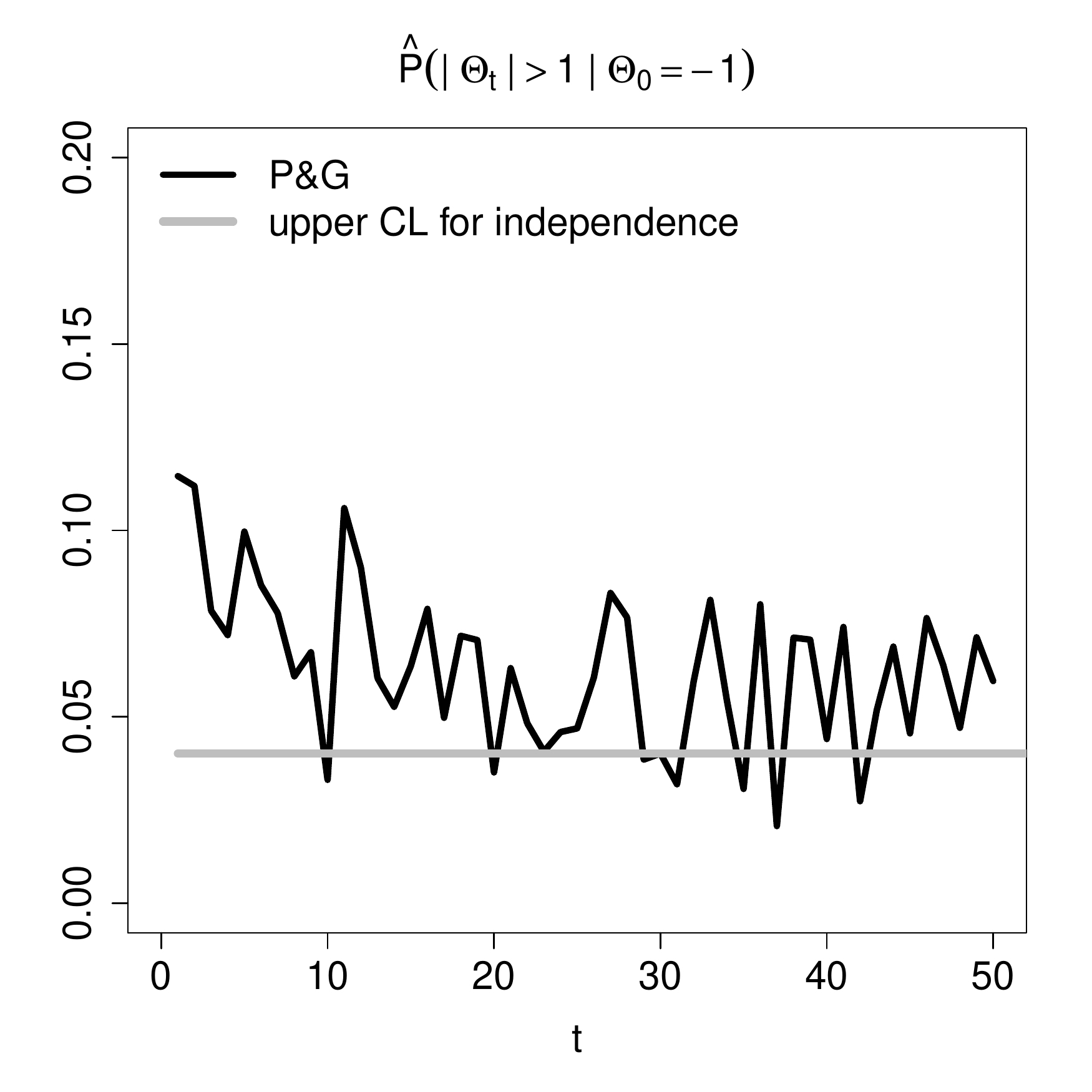}\\
\includegraphics[width=0.4\textwidth]{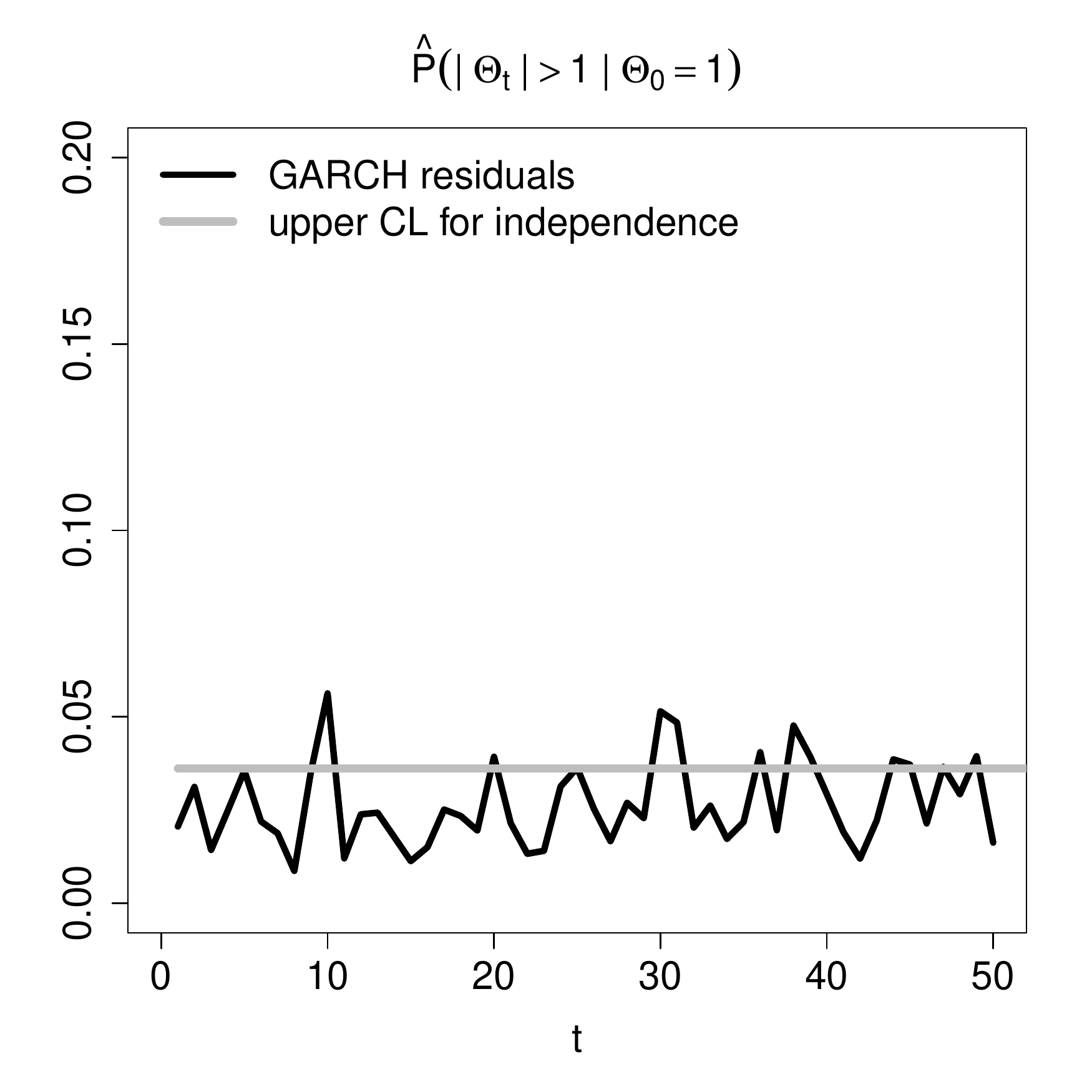}
&
\includegraphics[width=0.4\textwidth]{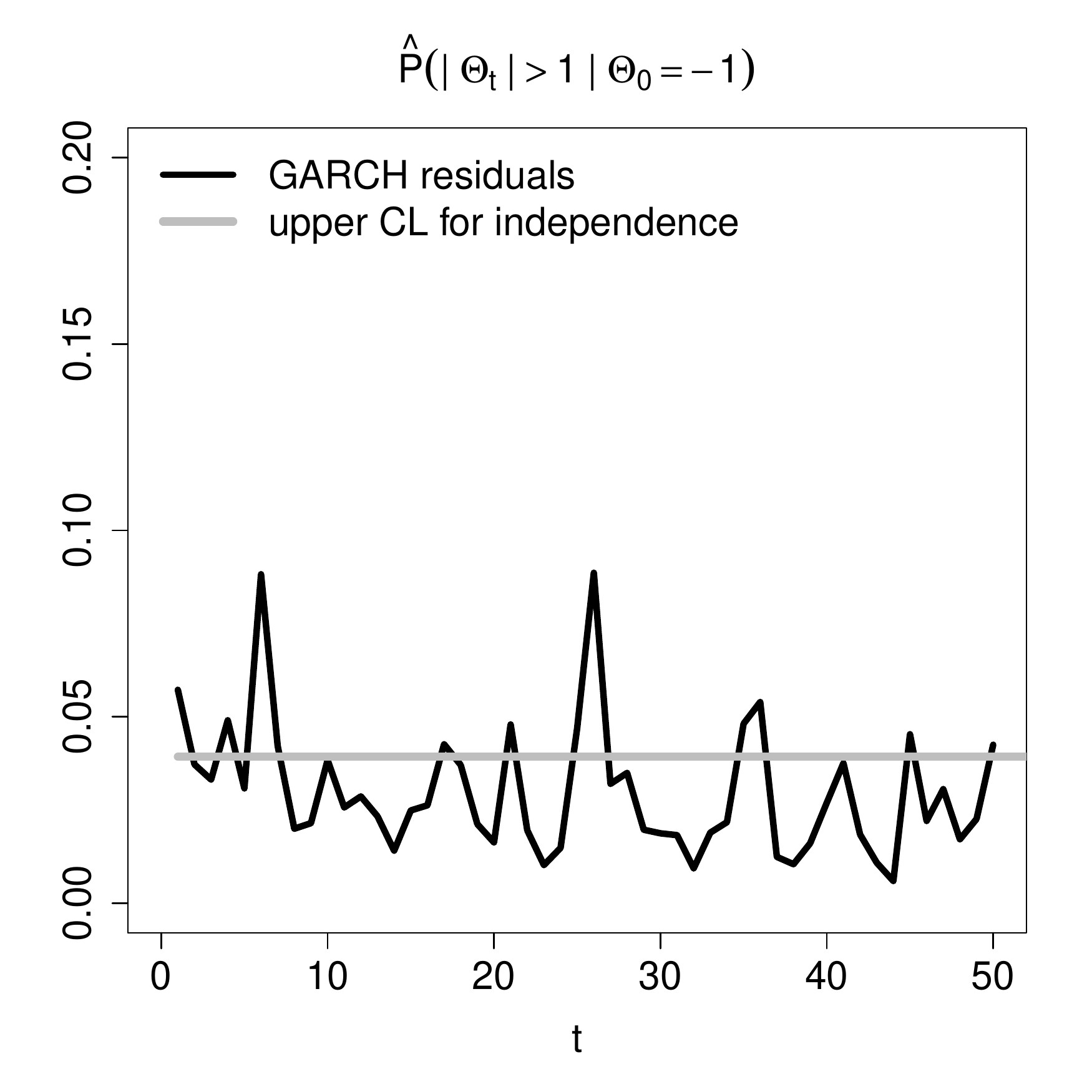}\\
\includegraphics[width=0.4\textwidth]{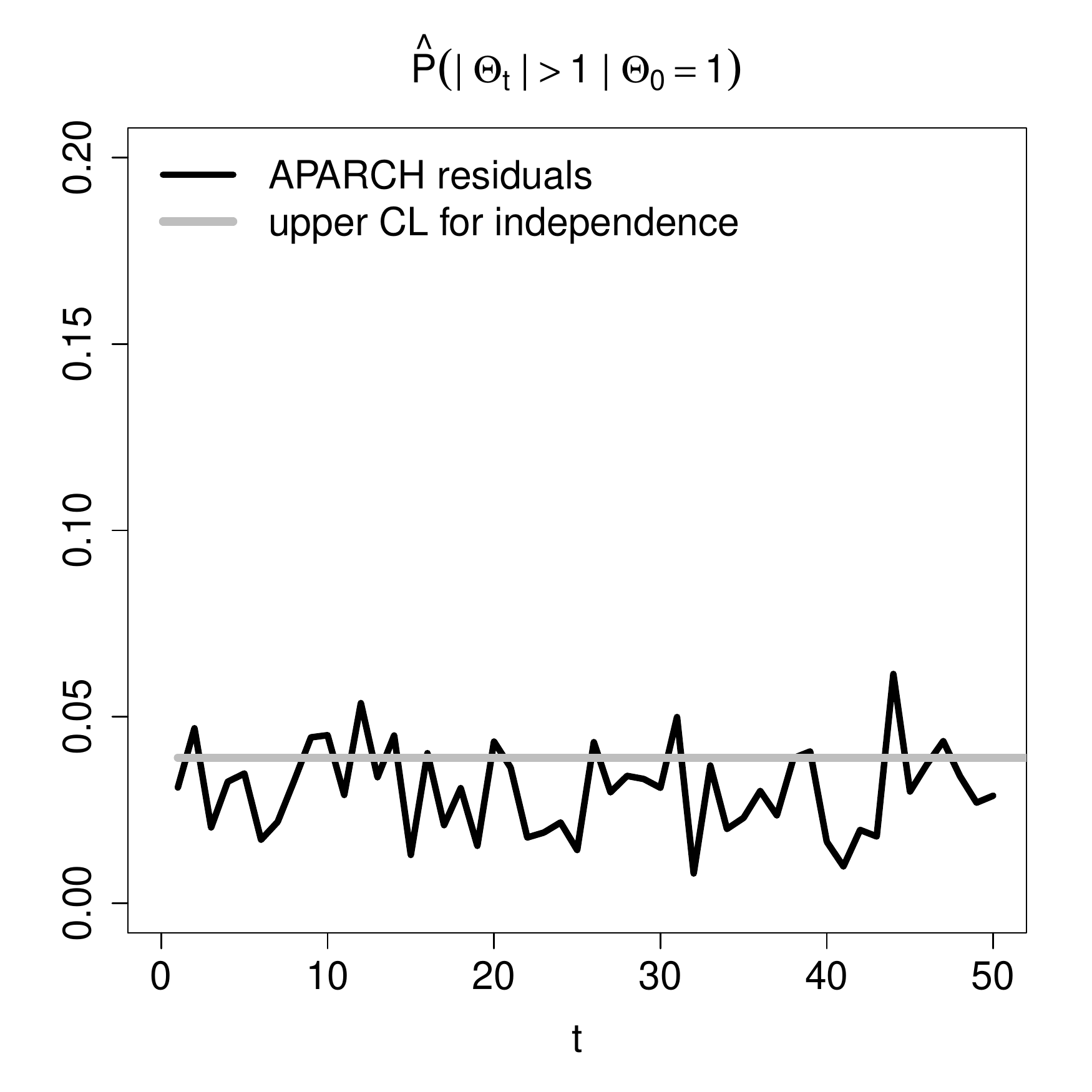}
&
\includegraphics[width=0.4\textwidth]{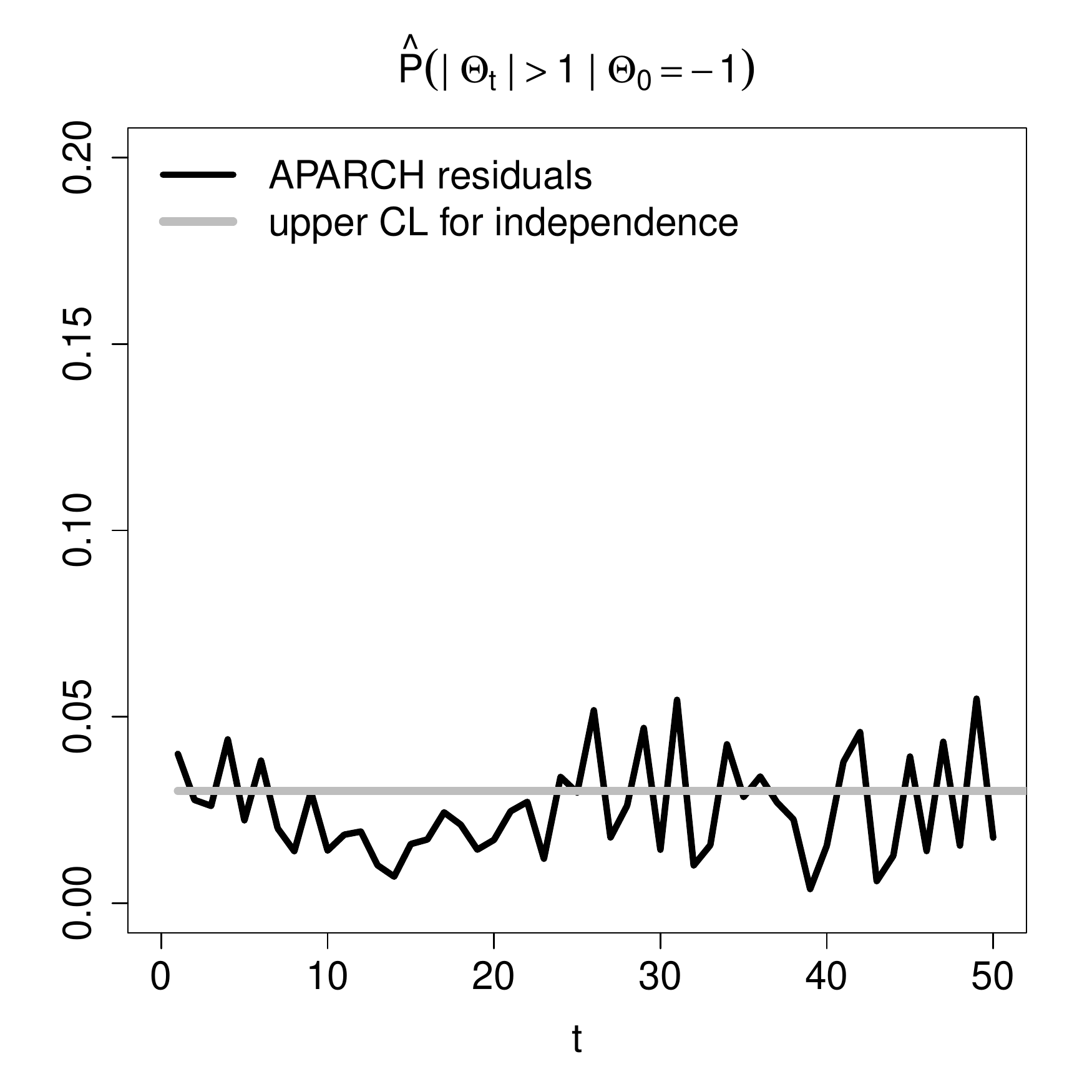}
\end{tabular}
\end{center}
\caption{\label{plot:app3} \footnotesize Sample spectral tail process (black line) for P\&G daily log-returns (top), GARCH(1,1) residuals (middle), and APARCH(1,1) residuals (bottom) based on the backward estimator. Plots in the first column represent conditioning on a positive shock whereas in the second column  one conditions on a negative shock. The horizontal gray lines correspond to the empirical $80\%$ quantile of the backward estimator under independence obtained from $10\,000$ simulations.}
\end{figure} 

\section{Proofs}
\label{sec:appendix}
\begin{proof}[Proof of Lemma \ref{lem:tc:T}]
To prove \eqref{eq:Tt:tc}, apply the time-change formula \eqref{eq:timechange} with $s=t=0$, $i=-h$, and $f(y_0)=\1(y_0\leq x) - \1(0\leq x)$ to see that
\[
  \prob{\Theta_h \leq x}
  - \1(0 \leq x)
  =
  \expec{\abs{\Theta_{-h}}^{\alpha}\1(\Theta_0/\abs{\Theta_{-h}}\leq x)}
  -
  \1(0 \leq x ) \expec{\abs{\Theta_{-h}}^{\alpha}}.
\]
For $x\geq 0$ in \eqref{eq:At:tc}, apply the time-change formula \eqref{eq:timechange} with $s=-h$, $t=0$, $i=-h$ and $f(y_{-h},\dots,y_0)=\1(y_0> x,\;y_{-h}=1)$ to get
\[\prob{\Theta_h > x,\; \Theta_0=1}=\expec{|\Theta_{-h}|^\alpha\,\1(\Theta_0/|\Theta_{-h}|>x,\;\Theta_{-h}>0)}
 = \expec{    \Theta_{-h}^{\alpha} \,
    \1(1/\Theta_{-h} > x,\;\Theta_0=1)},
\]
whereas for $x < 0$, take $f(y_{-h},\dots,y_0)=\1(y_0 \leq x,\;y_{-h}=1)$ to obtain
\[\prob{\Theta_h \leq x,\; \Theta_0=1}=\expec{
    \Theta_{-h}^{\alpha} \,
    \1(-1/\Theta_{-h} \leq x,\;\Theta_{-h}>0,\;\Theta_0=-1)}.
\]
Similarly, in \eqref{eq:Bt:tc} choose $f(y_{-h},\dots,y_0)=\1(y_0> x,\;y_{-h}=-1)$ and $f(y_{-h},\dots,y_0)=\1(y_0 \leq x,\;y_{-h}=-1)$ for $x\geq 0$ and $x < 0$, respectively.
\end{proof}

Next we turn to the asymptotic normality of the forward and backward estimators. Recall the definition of $X_{n,i}$ in \eqref{eq:Xni}. Consider the empirical process
\begin{equation}
\label{eq:tZn}
  \tilde{Z}_n(\psi) := (nv_n)^{-1/2} \sumIN \bigl( \psi(X_{n,i})-\expec{\psi(X_{n,i})} \bigr),
\end{equation}
where $\psi$ is  one of the following functions:
\begin{align}
\phi_0\left(y_{-\tilde{t}},\ldots,y_0,\ldots,y_{\tilde{t}}\right)&=\log^+(y_0), \nonumber\\
\phi_1\left(y_{-\tilde{t}},\ldots,y_0,\ldots,y_{\tilde{t}}\right)&=\1( y_0>1), \nonumber\\
\phi_{2,x}^t\left(y_{-\tilde{t}},\ldots,y_0,\ldots,y_{\tilde{t}}\right)&=\1(y_t/y_0 >x,\; y_0>1), \nonumber\\
\phi_{3,x}^t\left(y_{-\tilde{t}},\ldots,y_0,\ldots,y_{\tilde{t}}\right)
&=\left(y_{-t}/y_0\right)^\alpha\1(y_0/y_{-t} >x,  \;y_0>1)  \label{eq:functions}
\end{align}
for $\abs{t} \in \{1,\ldots,\tilde t\}$ and $x\ge 0$. The asymptotic behavior of $\tilde Z_n$ can be derived from more general results by \cite{drees2010limit}.

\begin{proposition}
\label{prop:procconv}
Let $(X_t)_{t\in\ZZ}$ be a non-negative, stationary, regularly varying time series with tail process $(Y_t)_{t \in \ZZ}$. Assume that conditions \ref{con:A}, \ref{con:B} and \ref{con:C'}  are fulfilled for some $x_0\ge 0$. Then, for all $y_0\in [x_0,\infty)\cap(0,\infty)$, the sequence of processes
\begin{equation*}
  \bigg(
    \tilde Z_n(\phi_0), \,
    \tilde Z_n(\phi_1), \,
    \left[
      (\tilde Z_n(\phi_{2,x}^t))_{x \in [x_0,\infty)}, \,
      (\tilde Z_n(\phi_{3,y}^t))_{y \in [y_0,\infty)}
    \right]_{\abs{t} \in \{1,\ldots,\tilde{t}\}}
  \bigg)
\end{equation*}
converges weakly to a centered Gaussian process $Z$ with covariance function given by
\begin{equation}
\label{eq:cov_emp_pr}
 \cov{Z(\psi_1),Z(\psi_2)}
 =
 \sum_{j=-\infty}^\infty
 \expec{
  \psi_1(Y_{-\tilde t},\ldots,Y_{\tilde t}) \,
  \psi_2(Y_{j-\tilde t},\ldots,Y_{j+\tilde t})
 }
 =:
 c(\psi_1,\psi_2)
\end{equation}
for all $\psi_1,\psi_2\in\left\{ \phi_0,\phi_1,\phi_{2,x}^t,\phi_{3,y}^t\mid x\ge x_0, \, y\ge y_0, \, |t|\in\{1,\ldots,\tilde t\}\right\}$.
\end{proposition}

The weak convergence statements in Proposition~\ref{prop:procconv} hold in the space of bounded functions on $\big\{ \phi_0,\phi_1,\phi_{2,x}^t,\phi_{3,y}^t\mid x\ge x_0, y\ge y_0,$ $|t|\in\{1,\ldots,\tilde t\}\big\}$ equipped with the supremum norm; see \citet[Section~1.5]{vdVW96} for details.

\begin{proof}[Proof of Proposition~\ref{prop:procconv}]
  One can argue similarly as in the proof of Proposition~B.1 of \cite{DreesSegersWarchol2015}, because the asymptotic equicontinuity of the process can be established for each $t$ separately. Note that the discussion in \cite{drees2016correct} shows that part~(ii) of condition~(B) of  \cite{DreesSegersWarchol2015} is not needed.

  By stationarity, the covariance of $Z(\psi_1)$ and $Z(\psi_2)$ is obtained as the limit of
   \begin{equation*}
  \frac 1{r_nv_n} \expec{\sum_{i=1}^{r_n}\psi_1(X_{n,i})\sum_{j=1}^{r_n}\psi_2(X_{n,j})}
   =  \frac 1{v_n} \sum_{k=-r_n+1}^{r_n-1} \Big(1-\frac{|k|}{r_n}\Big) \expec{\psi_1(X_{n,0}) \, \psi_2(X_{n,k})}.
  \end{equation*}
  This sum can be shown to converge to $c(\psi_1,\psi_2)$ using Pratt's lemma and Condition~\ref{con:C'}, as in \cite{DreesSegersWarchol2015}.
\end{proof}

\begin{remark}
  The covariances can be expressed in terms of the spectral tail process. For example,
  \begin{eqnarray*}
    c(\phi_{3,x}^t,\phi_0) & = & \sum_{j=-\infty}^\infty
    \expec{\Theta_{-t}^\alpha \, \1(1/\Theta_{-t}>x) \log^+(Y_0\Theta_j)}\\
    & = & \sum_{j=-\infty}^\infty
    \expec{\Theta_{-t}^\alpha \, \1(1/\Theta_{-t}>x) \big(\Theta_j^\alpha\wedge 1\big) \big(\log^+\Theta_j+\alpha^{-1}\big)}.
  \end{eqnarray*}
  Here we have used that $Y_0$ is independent of $(\Theta_s)_{s \in \ZZ}$ with distribution  $\prob{Y_0 > y} = y^{-\alpha}$ for $y \ge 1$.
\end{remark}

Theorem  \ref{theo:asnormest} and Corollary \ref{corol:probcenter} can now be proved in the same way as Theorem~4.5 in \cite{DreesSegersWarchol2015}. We omit the details, which can also be inferred from the more involved discussion of the bootstrap estimator below.

\cite{drees2015bootstrap}  has shown that under roughly the same conditions as used by \cite{drees2010limit}, conditionally on the data, the following bootstrap version of the empirical process $\tilde Z_n$ has the same asymptotic behavior as $\tilde Z_n$:
\begin{equation}
  \label{eq:Znxi}
  Z_{n,\xi}(\psi)
  :=
  (nv_n)^{-1/2} \sum_{j=1}^{m_n}  \xi_j \sum_{i\in I_j} \bigl( \psi(X_{n,i})-\expec{\psi(X_{n,i})} \bigr),
\end{equation}
with $I_j:=\{(j-1)r_n+1,\ldots,jr_n\}$ and $m_n:=\lfloor n/r_n \rfloor$.
In what follows, the symbol $\operatorname{E}_\xi$ denotes the expectation w.r.t.\ $\xi=(\xi_j)_{j\in\NN}$, i.e., the expectation conditionally on $(X_{n,i})_{1\le i\le n}$. Moreover, let $BL_1$ denote the set of all functions $g:\RR^{4\tilde{t}+2}\to \RR$ such that $\sup_{z\in\RR^{4\tilde{t}+2}}|g(z)|\le 1$ and $\abs{g(z_1)-g(z_2)} \le \lVert z_1-z_2 \rVert$ for all $z_1,z_2\in\RR^{4\tilde{t}+2}$.

\begin{proposition}
\label{prop:bootprocconv}
  Suppose that $(X_t)_{t\in\ZZ}$ is a non-negative, stationary, regularly varying time series and that the conditions \ref{con:A}, \ref{con:B} and \ref{con:C'}  are fulfilled for some $x_0\ge 0$. Then, for all $x\ge x_0$ and all $y_0\in [x_0,\infty)\cap(0,\infty)$, one has
  \begin{align*}
\Big( & Z_{n,\xi}(\phi_0),Z_{n,\xi}(\phi_1), \left[Z_{n,\xi}(\phi_{2,x}^t), Z_{n,\xi}(\phi_{3,y}^t)\right]_{|t|\in\{1,\ldots,\tilde{t}\}}\Big)
\\
 & \dto
\left( Z(\phi_0),Z(\phi_1), \left[Z(\phi_{2,x}^t), Z(\phi_{3,y}^t)\right]_{|t|\in\{1,\ldots,\tilde{t}\}}\right)
\end{align*}
with $Z$ as defined in Theorem \ref{theo:asnormest}.
Moreover,
\begin{align} \label{eq:bootconv}
 \sup_{g\in BL_1} \bigg|\operatorname{E}_\xi & g\left( Z_{n,\xi}(\phi_0),Z_{n,\xi}(\phi_1), \left[Z_{n,\xi}(\phi_{2,x}^t), Z_{n,\xi}(\phi_{3,y}^t)\right]_{|t|\in\{1,\ldots,\tilde{t}\}}\right) \nonumber \\
 & - \operatorname{E}g\left( Z(\phi_0),Z(\phi_1), \left[Z(\phi_{2,x}^t), Z(\phi_{3,y}^t)\right]_{|t|\in\{1,\ldots,\tilde{t}\}}\right) \bigg|\to 0
\end{align}
in probability.
\end{proposition}

Proposition~\ref{prop:bootprocconv} follows immediately from \citet[Theorem~2.1]{drees2015bootstrap}, because in the proof of Proposition~\ref{prop:procconv} (cf.\ the proof of Proposition~B.1 of \cite{DreesSegersWarchol2015}) it is shown that the assumptions of \citet[Theorem~2.1]{drees2015bootstrap} follow from the conditions of Proposition~\ref{prop:bootprocconv}.

Now we are ready to prove the consistency of the multiplier block bootstrap procedure.

\begin{proof}[Proof of Theorem \ref{theo:bootstrap}]
We only prove consistency of the bootstrap version of the backward estimator, as the proof for the forward  estimator is considerably simpler.
  For simplicity, we assume that $n=m_nr_n$. Let
  \[
    \alpha_n
    := \frac 1{\expec{\log^+(X_0/u_n)\mid X_0>u_n}}
    = \frac{v_n}{\expec{ \phi_0(X_{n,1}) }}
    .
  \]
  Recall $\tilde{Z}_n$ and $Z_{n,\xi}$ in \eqref{eq:tZn} and \eqref{eq:Znxi} respectively, recall $I_j = \{(j-1)r_n+1,\ldots,jr_n\}$, and recall $\hat{\alpha}_n$ and $\hat{\alpha}_n^*$ in \eqref{Hill} and \eqref{eq:Hill:boot}, respectively. Then
  \begin{eqnarray*}
    (nv_n)^{1/2} (\hat\alpha_n^*-\hat\alpha_n)
     & = & (nv_n)^{1/2} \frac{\sum_{j=1}^{m_n}\xi_j \sum_{i\in I_j} \1(X_i>u_n)-\hat\alpha_n \sum_{j=1}^{m_n}\xi_j \sum_{i\in I_j} \log^+(X_i/u_n)}{\sum_{j=1}^{m_n}(1+\xi_j) \sum_{i\in I_j} \log^+(X_i/u_n)}\\
     & = & \frac{Z_{n,\xi}(\phi_1)-\hat\alpha_n Z_{n,\xi}(\phi_0)+(r_nv_n)^{1/2} m_n^{-1/2}\sum_{j=1}^{m_n}\xi_j(1-\hat\alpha_n/\alpha_n)}{
     \alpha_n^{-1} (1+m_n^{-1}\sum_{j=1}^{m_n}\xi_j)+(nv_n)^{-1/2}\{\tilde Z_n(\phi_0)+Z_{n,\xi}(\phi_0)\}}.
  \end{eqnarray*}
 Since $m_n^{-1/2}\sum_{j=1}^{m_n}\xi_j$ and $\tilde Z_n$ are stochastically bounded and $\hat\alpha_n\to\alpha$ in probability, the assumptions $nv_n \to \infty$, $r_nv_n\to 0$, and $\alpha_n\to \alpha$, $m_n^{-1/2}$  and Proposition \ref{prop:bootprocconv} ensure that
  \begin{equation} \label{eq:boothillerr}
  (nv_n)^{1/2} (\hat\alpha_n^*-\hat\alpha_n)=\alpha Z_{n,\xi}(\phi_1)-\alpha^2 Z_{n,\xi}(\phi_0)+o_P(1),
   \end{equation}
   which converges weakly to $\alpha Z(\phi_1)-\alpha^2 Z(\phi_0)$. Moreover, conditionally on the data, it converges to the same limit weakly in probability in the sense of \eqref{eq:bootconv}.

 Next, recall $\CDFbTe{y}$ and $\hat{F}^{*({\mathrm{b}},\Theta_t)}_{n}(y)$ in \eqref{backward_Tte} and \eqref{eq:defbootbackward}, respectively. For $y > 0$, we have
  \[
    \left( 1 - \CDFbTe{y} \right)
    \sum_{i=1}^n \1(X_i > u_n)
    =
    \sum_{i=1}^n (X_{i-t}/X_i)^{\hat{\alpha}_n} \, \1( X_i/X_{i-t} > y, \, X_i > u_n ).
  \]
  It follows that
  \begin{multline*}
    \CDFbTe{y}-\hat{F}^{*({\mathrm{b}},\Theta_t)}_{n}(y) \\
    \shoveleft{
    = \Bigg[ \sum_{i=1}^n \left(\left(\frac{X_{i-t}}{X_i}\right)^{\hat\alpha_n^*}
       -\left(\frac{X_{i-t}}{X_i}\right)^{\hat\alpha_n}\right)
     \1(X_i/X_{i-t}>y,X_i>u_n)
    }
    \\
    \shoveleft{\qquad
    + \sum_{j=1}^{m_n}\xi_j\sum_{i\in I_j}\left(\frac{X_{i-t}}{X_i}\right)^{\hat\alpha_n^*}\1(X_i/X_{i-t}>y,X_i>u_n)
    }
    \\
    - \{1-\CDFbTe{y}\}\sum_{j=1}^{m_n}\xi_j\sum_{i\in I_j}\1(X_i>u_n)\Bigg]
    \bigg/
     \left[\sum_{j=1}^{m_n}(1+\xi_j)\sum_{i\in I_j}\1(X_i>u_n)\right].
  \end{multline*}
  For any pair $(\underline{\alpha}, \overline{\alpha})$ such that $0<\underline{\alpha}<\alpha<\overline{\alpha}$, there exists a constant $0<C<\infty$ such that for all $\tilde \alpha\in[\underline{\alpha},\overline{\alpha}]$ and, for suitable constants $\lambda=\lambda(\tilde\alpha)\in(0,1)$, we have, on the event $\{ X_i / X_{i-t} > y_0 \}$,
  \begin{eqnarray*}
    \lefteqn{\left|\left(\frac{X_{i-t}}{X_i}\right)^{\tilde\alpha}-\left(\frac{X_{i-t}}{X_i}\right)^{\alpha}
     - \left(\frac{X_{i-t}}{X_i}\right)^{\alpha} \log\left(\frac{X_{i-t}}{X_i}\right)(\tilde\alpha-\alpha)\right|}\\
     & = & \frac 12 \left(\frac{X_{i-t}}{X_i}\right)^{\alpha+\lambda(\tilde\alpha-\alpha)} \log^2\left(\frac{X_{i-t}}{X_i}\right)(\tilde\alpha-\alpha)^2
     \le  C (\tilde\alpha-\alpha)^2.
  \end{eqnarray*}
  Hence
  \begin{eqnarray}
    \lefteqn{\CDFbTe{y}-\hat{F}^{*({\mathrm{b}},\Theta_t)}_{n}(y)} \nonumber\\
     & = &
     \Bigg[ \sum_{i=1}^n \left(\frac{X_{i-t}}{X_i}\right)^{\alpha} \log\left(\frac{X_{i-t}}{X_i}\right)(\hat\alpha_n^*-\hat\alpha_n) \, \1(X_i/X_{i-t}>y, \, X_i>u_n)
      \nonumber\\
     & &
     +
     \sum_{j=1}^{m_n} \xi_j
     \sum_{i\in I_j}
     \left\{
	\left(\frac{X_{i-t}}{X_i}\right)^{\alpha}
	+
	\left(\frac{X_{i-t}}{X_i}\right)^{\alpha}
	\log \left(\frac{X_{i-t}}{X_i}\right)
	(\hat\alpha_n^*-\alpha)
     \right\}
     \1(X_i/X_{i-t}>y, \, X_i>u_n)  \nonumber\\
     & & - ( 1-\CDFbTe{y} ) \sum_{j=1}^{m_n}\xi_j\sum_{i\in I_j}\1(X_i>u_n)+R_n(y)\Bigg]\bigg/
     \left[\sum_{j=1}^{m_n}(1+\xi_j)\sum_{i\in I_j}\1(X_i>u_n)\right]
       \label{eq:booterr}
  \end{eqnarray}
  with
  \begin{eqnarray*}
   \abs{ R_n(y) }
   & \le &
   C(\hat\alpha_n^*-\hat\alpha_n)^2\sum_{i=1}^n \1(X_i/X_{i-t}>y, \, X_i>u_n) \\
   & &
   \mbox{} + C(\hat\alpha_n^*-\alpha)^2
   \sum_{j=1}^{m_n} \abs{ \xi_j }
   \sum_{i\in I_j} \1(X_i/X_{i-t}>y, \, X_i>u_n) \\
   & = &
   O_P \bigl( (nv_n)^{-1}nv_n + (nv_n)^{-1}m_nr_nv_n \bigr)
   =O_P(1), \qquad n \to \infty.
  \end{eqnarray*}

  Consider the function
  \[
    \phi_{4,x}^t\left(y_{-\tilde{t}},\ldots,y_0,\ldots,y_{\tilde{t}}\right)
    =
    \left(y_{-t}/y_0\right)^\alpha
    \log\left(y_{-t}/y_0\right)
    \1(y_0/y_{-t} >x, \, y_{-t}>0, \, y_0>1).
  \]
  One may show as in the proof of Proposition \ref{prop:procconv} that $\tilde Z_n(\phi_{4,y}^t)$ and $Z_{n,\xi}(\phi_{4,y}^t)$ both converge weakly to $Z(\phi_{4,y}^t)$. In particular, as $n \to \infty$,
  \begin{eqnarray*}
  \lefteqn{ (nv_n)^{-1} \sum_{i=1}^n \left(\frac{X_{i-t}}{X_i}\right)^{\alpha} \log\left(\frac{X_{i-t}}{X_i}\right)\1(X_i/X_{i-t}>y,X_i>u_n)}\\
  & = &
  \expec{
    \left(\frac{X_{-t}}{X_0}\right)^{\alpha}
    \log\left(\frac{X_{-t}}{X_0}\right)
    \1(X_0/X_{-t}>y)
    \,\Big|\,
    X_0>u_n
  }
  + O_P\bigl((nv_n)^{-1/2}\bigr) \\
  & \to &
  \expec{ \Theta_{-t}^\alpha \log(\Theta_{-t}) \, \1(1/\Theta_{-t}>y)} \\
  & = &
  - \expec{\log(\Theta_{t}) \, \1(\Theta_{t}>y)},
  \end{eqnarray*}
  where the last step follows from the time-change formula~\eqref{eq:timechange} applied with $f(y_0) = - \log(y_0) \, \1(y_0 > y)$ and $(-t, 0, -t)$ instead of $(s, t, i)$. Therefore
  \begin{eqnarray}
    \lefteqn{\sum_{i=1}^n \left(\frac{X_{i-t}}{X_i}\right)^{\alpha} \log\left(\frac{X_{i-t}}{X_i}\right)(\hat\alpha_n^*-\hat\alpha_n)\1(X_i/X_{i-t}>y,X_i>u_n)}
    \nonumber
    \\
    & = & -(nv_n)^{1/2}\big(\expec{\log(\Theta_{t}) \, \1(\Theta_{t}>y)}+o_P(1)\big)(nv_n)^{1/2}(\hat\alpha_n^*-\hat\alpha_n).
     \label{eq:l1}
  \end{eqnarray}
  Likewise, one can conclude that
  \begin{align*}
   & (nv_n)^{-1/2} \sum_{j=1}^{m_n}\xi_j\sum_{i\in I_j} \left(\frac{X_{i-t}}{X_i}\right)^{\alpha}\log\left(\frac{X_{i-t}}{X_i}\right)
     \1(X_i/X_{i-t}>y,X_i>u_n)\\
    & =  Z_{n,\xi}(\phi_{4,y}^t)+ (rv_n)^{1/2} m_n^{-1/2} \sum_{j=1}^{m_n}\xi_j \expec{\Big(\frac{X_{-t}}{X_0}\Big)^\alpha\log\left(\frac{X_{-t}}{X_0} \right) \1(X_0/X_{-t}>y)\,\Big|\, X_0>u_n}\\
   &  =  O_P(1).
  \end{align*}
   As a consequence,
  \begin{align}
    \lefteqn{
    \sum_{j=1}^{m_n} \xi_j
    \sum_{i\in I_j}
    \left\{
      \left(\frac{X_{i-t}}{X_i}\right)^{\alpha}
      +
      \left(\frac{X_{i-t}}{X_i}\right)^{\alpha}
      \log\left(\frac{X_{i-t}}{X_i}\right)
      (\hat\alpha_n^*-\alpha)
    \right\}
    \1(X_i/X_{i-t}>y,X_i>u_n)
    } \nonumber
    \\
    & =  (nv_n)^{1/2}Z_{n,\xi}(\phi_{3,y}^t)+\sum_{j=1}^{m_n}\xi_j r_nv_n \big(\expec{\Theta_{-t}^\alpha\1(1/\Theta_{-t}>y)}+o(1)\big) + O_P(1) \nonumber
    \\
    & =  (nv_n)^{1/2}\left(Z_{n,\xi}(\phi_{3,y}^t)+O_P \bigl( (r_nv_n)^{1/2} \bigr)+O_P\big((nv_n)^{-1/2}\big)\right).
    \hspace*{4cm}
    \label{eq:l2}
  \end{align}

  Moreover,
  we find, as $r_n v_n \to 0$ and $\sum_{j=1}^{m_n}\xi_j=O_P(m_n^{1/2})$, that
  \begin{eqnarray}
    \sum_{j=1}^{m_n}\xi_j\sum_{i\in I_j}\1(X_i>u_n)
    &=&
    (n v_n)^{1/2} \, Z_{n,\xi}(\phi_1) + r_nv_n \sum_{j=1}^{m_n}\xi_j\nonumber \\
    &=&
    (n v_n)^{1/2} \bigl( Z_{n,\xi}(\phi_1) + o_P(1) \bigr), \qquad n \to \infty. \label{eq:l3}
  \end{eqnarray}
  The denominator of \eqref{eq:booterr} equals $nv_n+O_P((nv_n)^{1/2})$. Combining \eqref{eq:booterr}--\eqref{eq:l3} and \eqref{eq:boothillerr} yields
    \begin{eqnarray*}
    \lefteqn{(nv_n)^{1/2}\big(\CDFbTe{y} - \hat{F}^{*({\mathrm{b}},\Theta_t)}_{n}(y)\big)}\\
     & = &
     -\expec{\log(\Theta_{t}) \, \1(\Theta_{t}>y)}(nv_n)^{1/2}(\hat\alpha_n^*-\hat\alpha_n)
     +
     Z_{n,\xi}(\phi_{3,y}^t)
     -
     ( 1-\CDFbTe{y} ) \, Z_{n,\xi}(\phi_1)
     + o_P(1) \\
     & = &
     Z_{n,\xi}(\phi_{3,y}^t)
     -
     \bar{F}^{(\Theta_t)}(y) \, Z_{n,\xi}(\phi_1)
     -
     \expec{\log( \Theta_{t}) \, \1(\Theta_{t}>y)}
     \left(\alpha Z_{n,\xi}(\phi_1)-\alpha^2 Z_{n,\xi}(\phi_0)\right)
     + o_P(1).
  \end{eqnarray*}
  Now the assertion is a direct consequence of Proposition \ref{prop:bootprocconv} and Theorem \ref{theo:asnormest}.
\end{proof}

\section*{Acknowledgements}

The authors wish to thank the editors and the referees for their careful reading and for various constructive comments and useful suggestions. J.\ Segers gratefully acknowledges funding by contract ``Projet d'Act\-ions de Re\-cher\-che Concert\'ees'' No.\ 12/17-045 of the ``Communaut\'e fran\c{c}aise de Belgique'' and by IAP research network Grant P7/06 of the Belgian government (Belgian Science Policy). The research of M.~Warcho\l{} was funded by a PhD grant of the ``Fonds de la Recherche Scientifique'' (F.R.S.-FNRS). H.~Drees was partially supported by DFG research grant JA 2160/1. R.~Davis was supported in part by ARO MURI grant W911NF-12-1-0385.

\small
\bibliographystyle{chicago}
\bibliography{mybib}

\end{document}